\documentclass[reqno, 12pt]{amsart}
\usepackage{geometry}
\usepackage{amsmath,amssymb}
\usepackage{amsfonts}
\usepackage{eucal}
\usepackage{amsthm}
\usepackage{mathrsfs}
\numberwithin{equation}{section}
\usepackage{comment}

\newcommand{\jap}[1]{\langle #1 \rangle}

\def\a{\alpha}
\def\b{\beta}
\def\c{\gamma}
\def\d{\delta}
\def\e{\varepsilon}
\def\f{\varphi}
\def\g{\psi}

\def\i{\mbox{\raisebox{.5ex}{$\chi$}}}
\def\k{\kappa}
\def\l{\lambda}
\def\m{\mu}

\def\s{\sigma}
\def\t{\tau}
\def\x{\xi}
\def\y{\eta}
\def\z{\zeta}
\def\th{\theta}

\renewcommand{\L}{\Lambda}
\renewcommand{\O}{\Omega}

\newcommand{\Op}{\mathrm{Op}}

\def\th{\hat{t}}
\def\yh{\hat{y}}

\def\re{\mathbb{R}}

\def\pa{\partial}

\renewcommand{\Re}{\text{{\rm Re}\;}}
\renewcommand{\Im}{\text{{\rm Im}\;}}

\newcommand{\supp}{\text{{\rm supp}\;}}

\newtheorem{thm}{Theorem}[section]
\newtheorem{lem}[thm]{Lemma}
\newtheorem{prop}[thm]{Proposition}
\newtheorem{cor}[thm]{Corollary}

\theoremstyle{definition}
\newtheorem{defn}{Definition}
\newtheorem{ass}{Assumption}

\theoremstyle{remark}
\newtheorem{rem}[thm]{Remark}

\title[]%
{Limiting absorption principle and equivalence of Feynman propagators on asymptotically Minkowski spacetimes}
\author{Kouichi Taira}
\begin{document}
\begin{abstract}
In this paper, we shall show that the limiting absorption principle for the wave operator on the asymptotically Minkowski spacetime. This problem was previously considered by [A. Vasy, J. Spect. Theory,  10,439-461 , (2020)]. Here, we employ a more transparent tool, the Mourre theory and remove an additional condition which is imposed in his paper. Moreover, we also prove that the anti-Feynman propagator defined by G\'erard and Wrochna coincides with the outgoing resolvent.
\end{abstract}
\maketitle

\section{Introduction}

Let $g_0$ be a Minkowski metric on $\re^{1+n}$ and $g_0^{-1}$ be its dual metric:
\begin{align*}
g_0=-dx_1^2+dx_2^2+...+dx_{1+n}^2,\quad g_0^{-1}=-\pa_{x_1}^2+\pa_{x_2}^2+...+\pa_{x_{1+n}}^2=(g^{ij}_0)_{i,j=1}^n.
\end{align*}
We denote $\jap{x}=(1+|x|^2)^{\frac{1}{2}}$ and introduce the function space
\begin{align*}
S^k(\re^{n+1}):=\{a\in C^{\infty}(\re^{n+1})\mid |\pa_x^{\a}a(x)|\leq C\jap{x}^{k-|\a|}\},\quad \quad k\in \re.
\end{align*}

\begin{defn}
A Lorentzian metric $g$ on $\re^{n+1}$ is called asymptotically Minkowski if the inverse matrix $g^{-1}(x)=(g^{jk}(x))_{j,k=1}^n$ of $g(x)$ satisfies $g^{jk}- g_0^{jk}\in S^{-\m}(\re^{n+1})$ for some $\m>0$.

\end{defn}

In this paper, we consider the second order differential operator
\begin{align}\label{Pdef}
P=-\Box_g+V(x)=-\sum_{j,k=1}^{n+1}|g|^{-\frac{1}{2}}\pa_{x_j}(|g|^{\frac{1}{2}}g^{jk}(x)\pa_{x_k})+V(x),\quad V\in S^{-\m}(\re^{n+1})
\end{align}
where $g$ is an asymptotically Minkowski space and $|g(x)|=|\det g(x)|$. As is shown in \cite{NT} and \cite{V}, $P$ is essentially self-adjoint on $C_c^{\infty}(\re^{n+1})$ under Assumption \ref{nulltrapp} below (on the Hilbert space $L^2(\re^{n+1}; |g|^{\frac{1}{2}}dx)$). Here, we note that the symmetric operator $P$ on $L^2(\re^{n+1}; |g|^{\frac{1}{2}}dx)$ is unitaty equivalent to
\begin{align*}
P_1=|g|^{-\frac{1}{4}}P|g|^{\frac{1}{4}}
\end{align*}
on the usual $L^2$-space $L^2(\re^{n+1})$ and $P_1$ satisfies the assumption in \cite{NT}.
We denote their unique self-adjoint extensions by the same symbols $P$.

In Quantum Field Theory (QFT), it is important to seek special solutions for the Klein-Gordon equation $(P+m_0^2)u=f$ such as the advanced/retarded propagators and the Feynman/anti-Feynman propagators. The existence of the advanced/retarded propagators is classically known, however, the existence of the Feynman/anti-Feynman propagators had beed unknown for a long time.

In the exact Minkowski spacetime, a Feynman propagator is defined as the outgoing resolvent of the d'Alembert operator:
\begin{align*}
\frac{1}{\pa_{t}^2-\Delta_{y}+m_0^2-i0},\quad x=(t,y)\in \re\times \re^n,
\end{align*}
which is interpreted as a Fourier multiplier of the distribution $(-\x_1^2+\x_2^2+...+\x_{n+1}^2+m_0^2-i0)^{-1}$.
Natural questions are whether a Feynman propagator exists and whether such formula is true even on curved spacetimes. 
 In the recent advance of the microlocal analysis or a time evolution approach, the Feynman propagators are defined for various spacetimes (\cite{DeZ1}, \cite{DeZ2}, \cite{GW1}, \cite{GW2}, \cite{V}). However, a relationship between such Feynman propagators seems unknown although Feynman propagators coincide up to the smooth integral kernel by the earlier work of Duistermaat and H\"ormander \cite{DH}.

Vasy \cite{V} constructed a Feynman propagator on the asymptotically Minkowski spacetimes as the outgoing/incoming resolvents for $P+m_0^2$ by using the radial estimates, although its proof in detail is omitted and an additional assumption is imposed.
The purpose of this article is the following:
\begin{itemize}
\item We prove the limiting absorption principle for $P$ for non-zero energies via the Mourre theory which is more transparent in the operator theory. Moreover, we  remove the additional assumption in \cite{V} (see also \cite[Proposition 2.12]{DW}).
\item We show that the (anti-)Feynman propagator defined in \cite{GW1} and \cite{GW2} is written as the limit of the resolvent to the real line.
\end{itemize}
Moreover, we also show that the resolvents $(P\pm i)^{-1}$ preserve the Schwartz space $\mathcal{S}(\re^{n+1})$.

\subsection{Main results}

Set
\begin{align*}
\,\,p(x,\x)=\sum_{j,k=1}^{n+1}g^{jk}(x)\x_j\x_k,\quad p_0(\x)=\sum_{j,k=1}^{n+1}g^{jk}_0\x_j\x_k,\,\, P_0=\pa_t^2-\Delta_y.
\end{align*}
Let $(z(t, x, \x),\z(t, x, \x))$ denote the solution
to the Hamilton equation:
\begin{align*}
\frac{d}{dt} z(t,x,\x) =\frac{\pa p}{\pa \x}(z(t,x,\x),\z(t,x,\x)), \quad
\frac{d}{dt} \z(t) =-\frac{\pa p}{\pa x}(z(t,x,\x),\z(t,x,\x)), \quad
t\in\re
\end{align*}
with $z(0,x,\x)=x$, $\z(0,x,\x)=\x$.

\begin{ass}[Null non-trapping condition]\label{nulltrapp}
For each $(x,\x)\in p^{-1}(\{0\})\setminus \{\x=0\}$, we have $|z(t,x,\x)|\to \infty$ as $|t|\to \infty$.
\end{ass}

Our first main result is the following limiting absorption principle for $P$ away from the zero energy. 
\begin{thm}\label{lapwave}
Assume Assumption \ref{nulltrapp}. Let $s>1/2$ and $I\Subset \re\setminus \{0\}$ be an open interval. Then it follows that $\# I\cap \s_{pp}(P)$ is finite and that for $I'\Subset I\setminus \s_{pp}(P)$, we have
\begin{align*}
\sup_{z\in I'_{\pm}}\|\jap{x}^{-s}(P-z)^{-1}\jap{x}^{-s}\|_{B(L^2(\re^{n+1}))}<\infty,
\end{align*}
where
\begin{align*}
I'_{\pm}=\{z\in \mathbb{C}\mid \Re z\in I,\, \pm\Im z> 0\}.
\end{align*}
In particular, $P$ has absolutely continuous spectrum on $I'$. Moreover, $z\in I'_{\pm} \mapsto \jap{x}^{-s}(P-z)^{-1}\jap{x}^{-s}$ is H\"older continuous in $B(L^2(\re^{n+1}))$ and the limits
\begin{align*}
\jap{x}^{-s}(P-\l\mp i0)^{-1}\jap{x}^{-s}:=\lim_{\e\to 0,\,\e>0}\jap{x}^{-s}(P-\l\mp i\e)^{-1}\jap{x}^{-s}
\end{align*}
exist in $B(L^2(\re^{n+1}))$.

\end{thm}

\begin{rem}
The above theorem for the interval $I$ is essentially proved by Vasy \cite{V} and \cite[Proposition 2.12]{DW} under the non-trapping assumption on energy level $I$. We remove this additional assumption by proving the local compactness of $P$ (Proposition \ref{waveloccor}).
\end{rem}

To state our second result, we introduce the additional assumption which is imposed in \cite{GW1} and \cite{GW2}. We denote 
\begin{align*}
x=(t,y)\in \re\times \re^{n}.
\end{align*}

\begin{ass}\label{timefct}
There exists a time function $\tilde{t}$ such that $\tilde{t}-t\in C_c^{\infty}(\re^{n+1})$.
\end{ass}

\begin{rem}
As is proved in \cite[Proposition 4.3 $(ii)$]{GW1}, the space $C_c^{\infty}(\re^{n+1})$ in Assumption \ref{timefct} can be replaced by $S^{-\e}(\re^{n+1})$ with $\e>0$. Moreover, the authors in \cite{GW1} conjecture that the existence of such function follows from Assumption \ref{nulltrapp} and the asymptotically Minkowski structure of $g$.
\end{rem}

We also recall some notations used in \cite{GW1} and \cite{GW2}, which are also reviewed in subsection \ref{GWsusec}. We note that the convention of the Feynman propagator in \cite{GW1} and \cite{GW2} is opposite to the convention in physics. In this paper, we use the convention in physics. It is shown in \cite{GW1} that under Assumptions \ref{nulltrapp} and \ref{timefct} with $\m>1$, there exists a diffeomorphism $\chi:\re^{n+1}\to \re^{n+1}$ such that
\begin{align*}
(\chi^*g)_{(\hat{t},\hat{y})}=-\hat{c}(\hat{t},\hat{y})^2d\hat{t}^2+\sum_{i,j=1}^n\hat{h}_{ij}(\hat{t},\hat{y})d\hat{y}^id\hat{y}^j,\quad (\hat{t},\hat{y})\in \re^{n+1}.
\end{align*}
For $m\in \re$, let $\mathcal{Y}^m$ and $\mathcal{X}^m_{\overline{F}}$ be the function spaces as in \cite{GW1} (see also subsection \ref{GWsusec}). Then \cite[Theorem 1.1]{GW2} and its proof state that for $m_0>0$,
\begin{align*}
P+m_0^2:\mathcal{X}^m_{\overline{F}}\to \mathcal{Y}^m
\end{align*}
is boundedly invertible. We denote its inverse by $(P+m_0^2)_{\overline{F}}^{-1}$. Note that Lemma \ref{deffeoprop} and the definition of $\mathcal{Y}^m$ give a natural inclusion $\mathcal{S}(\re^{n+1})\subset \mathcal{Y}^m$. Thus the inverse $(P+m_0^2)_{\overline{F}}^{-1}$ is defined on the Schwartz space $\mathcal{S}(\re^{n+1})$.

The next theorem shows that the (anti-)Feynman propagator constructed in \cite{GW2} coincides with the outgoing resolvent of $P$.

\begin{thm}\label{Feyeq}
Assume Assumptions \ref{nulltrapp} and \ref{timefct} with $\m>1$.

\noindent$(i)$ For $m_0>0$, $-m_0^2\notin \s_{pp}(P)$, where $\s_{pp}(P)$ is the pure point spectrum of the self-adjoint operator $P$. In particular, the outgoing/incoming resolvent $(P+m_0^2\mp i0)^{-1}\in B(L^{2,s}(\re^{n+1}), L^{2,-s}(\re^{n+1}))$ exist for each $s>\frac{1}{2}$ by virtue of Theorem \ref{lapwave}.

\noindent$(ii)$ We have the identity
\begin{align*}
(P+m_0^2)_{\overline{F}}^{-1}=(P+m_0^2-i0)^{-1}\quad \text{on}\quad \mathcal{S}(\re^{n+1}).
\end{align*}

\end{thm}

\begin{rem}
By the density argument, the space $\mathcal{S}(\re^{n+1})$ in $(ii)$ can be improved by $L^{2,\frac{1}{2}+\e}(\re^{n+1})$ with $\e>0$.
\end{rem}

Finally, we shall state a result on a good mapping property of the resolvent $(P-i)^{-1}$, which is less motivative in QFT, however, is hopefully useful for the spectral analysis of $P$. We note that the operator $(P-i)^{-1}$ maps from $\mathcal{S}(\re^{n+1})$ to $L^2(\re^{n+1})$. Next theorem claims that the range of this operator is contained in $\mathcal{S}(\re^{n+1})$.

\begin{thm}\label{Acompro}
Assume Assumption \ref{nulltrapp} and let $z\in \mathbb{C}\setminus \re$. Then the resolvent $(P-z)^{-1}$ is a linear continuous operator on $\mathcal{S}(\re^{n+1})$. In particular, $(P-z)^{-1}$ can be extended to a linear continuous operator on $\mathcal{S}'(\re^{n+1})$.
\end{thm}

\subsection{Outline of the proof}

In this subsection, we give an outline of the proof for our main theorem.

\noindent\textbf{Limiting absorption principle} The proof of Theorem \ref{lapwave} is based on the Mourre theory \cite{Mo}. We can easily construct a self-adjoint operator $A$ satisfying the Mourre estimate:
\begin{align*}
\f(P) [P, iA]\f(P)\geq C\f(P)^2+K,
\end{align*}
where $K$ is a remainder term of the form $K=\f(P)R\f(P)+\Op S^{-\infty,-\infty}$ with $R\in \Op S^{0,-\m}$. So our task is to prove that the remainder term $K$ is actually a compact operator. For an elliptic operator (like the Laplacian $-\Delta$), it is easily done by the standard elliptic estimate. In our case, the lack of the ellipticity of our operator $P$ makes this problem more difficult. To prove the compactness of $K$, we shall deduce the local compactness property of $P$, that is, the compactness of $\chi(x)(P-i)^{-1}$ for $\chi\in C_c^{\infty}(\re^{n+1})$.
In fact, it is proved in Proposition \ref{waveloccor} that $(P-i)^{-1}$ maps from $L^2(\re^{n+1})$ to $H^{\frac{1}{2},-\frac{1}{2}-\e}(\re^{n+1})$ with $\e>0$. For explaining the idea of the proof of Proposition \ref{waveloccor}, we consider the Schr\"odinger propagator $\{e^{-isP}\}_{s\in \re}$ for the Hamiltonian $P$. Under a more restrictive condition: globally non-trapping condition, Chihara \cite[Theorem 1.2]{Ch} proves the local smoothing estimate for $e^{-isP}$: For $T,\e>0$, we have
\begin{align}\label{tlocsmo}
\|e^{-isP}u\|_{L^2([-T,T],H^{\frac{1}{2},-\frac{1}{2}-\e}(\re^{n+1}))}\leq C \|u\|_{L^2(\re^{n+1})}.
\end{align}
Using this estimate, we shall prove that $(P-i)^{-1}$ maps $L^2(\re^{n+1})$ to $H^{\frac{1}{2},-\frac{1}{2}-\e}(\re^{n+1})$. For simplicity, we assume $|g|=1$. For $f\in L^2(\re^{n+1})$, set $u=(P-i)^{-1}\in L^2(\re^{n+1})$, $f(s)=e^sf$ and $u(s)=e^su$. Then $u(s)$ satisfies the time-dependent Schr\"odinger equation $(i\pa_s-P)u(s)=-f(s)$ with $u(0)=u$. Hence the Duhamel formula and $(\ref{tlocsmo})$ imply
\begin{align*}
\|u(s)\|_{L^2([-T,T],H^{\frac{1}{2},-\frac{1}{2}-\e})}\leq& \|e^{-isP}u\|_{L^2([-T,T],H^{\frac{1}{2},-\frac{1}{2}-\e})}+\|\int_0^te^{-i(t-s)P}f(s)\|_{L^2([-T,T],H^{\frac{1}{2},-\frac{1}{2}-\e})}\\
\leq& C'\|u\|_{L^2}+C'\|f\|_{L^2}.
\end{align*}
Since $\|u(s)\|_{L^2([-T,T],H^{\frac{1}{2},-\frac{1}{2}-\e})}=C''\|u\|_{H^{\frac{1}{2},-\frac{1}{2}-\e}}$, we have $u\in H^{\frac{1}{2},-\frac{1}{2}-\e}(\re^{n+1})$. Actually, we only impose the null(not global)-trapping condition and hence cannot use the above argument. Instead, in Lemma \ref{escapenullnon}, we prove a microlocal alternative of \cite[Lemma 7.2]{Ch}, which is a key estimate for the proof of $(\ref{tlocsmo})$ in \cite{Ch}.

\noindent \textbf{Equivalence of the Feynman propagators} The proof of Theorem \ref{Feyeq} $(i)$ immediately follows from the invertibility of $P_0+m_0^2$ on $\mathcal{X}_F^m$ and the radial estimates. Hence we only focus on the part $(ii)$.
 To prove Theorem \ref{Feyeq} $(ii)$, it suffices to show that for $u\in \cap_{m\in \re}\mathcal{X}_{\overline{F}}^m$ with $f:=(P+m_0^2)u\in \mathcal{S}(\re^{n+1})$, we have $u=(P+m_0^2-i0)^{-1}f$. Essentially due to the work by G\'erard and Wrochna, $u$ has an asymptotic expansion of the form
\begin{align}\label{introasym}
u(t)=\begin{cases}
e^{-it\sqrt{-\Delta+m_0^2}}g_{+,-}+\text{l.o.t.} \quad \text{for}\quad  t>>1,\\
e^{it\sqrt{-\Delta+m_0^2}}g_{-,+}+ \text{l.o.t.}\quad \text{for}\quad -t>>1.
\end{cases}
\end{align}
for some $\e>0$ and $g_{+,-}, g_{-,+}\in H^m(\re^n)$. On the other hand, in order to prove $u=(P+m_0^2-i0)^{-1}f$, we only need to check a microlocal radiation condition which is analogue to \cite[Proposition 14]{Me}. Such microlocal radiation condition is written as follows:
\begin{align}\label{intromiccond}
u \in L^{2,-\frac{1}{2}+\e}(\re^{n+1})\,\, \text{microlocally away from the outgoing region defined in (\ref{outindef})},
\end{align}
which is in turn equivalent to the condition $u \in \mathcal{S}(\re^{n+1})$ microlocally in the incoming region by virtue of the propagation of singularities. The second key ingredient is the observation that at each microlocally incoming point $(x,\x)$, the time variable is sufficiently large, which follows from a formal calculation $p_0(\x)+m_0^2=0$, $|x|>>1$, $\frac{x\cdot \pa_{\x}p_0(\x)}{|x||\pa_{\x}p_0(\x)|}=-1 \Rightarrow |t|>>1$. From this observation, it turns out that the scattering theory for $|x|>>1$ can be replaced by the scattering theory for $|t|>>1$. Thus, we only need to prove $(\ref{intromiccond})$ for $|t|>>1$, which is easily checked by the asymptotic expansion $(\ref{introasym})$ and a standard elliptic parametrix construction. For a detail, see subsection \ref{subsecrad}.

Finally, we breifly explain why the condition $(\ref{intromiccond})$ is called the radiation condition. For the standard Laplacian $P=-\Delta=\Op(p_0)$ and the energy $\l>0$, Sommerfeld's radiation condition can be written by $\l^{-1}D_ru-u\in L^{2,-\frac{1}{2}+\e}$. The Euler formula gives $\l^{-1}D_r=\frac{x}{\l r}\cdot D_x$ and its principal symbol coincides with $\frac{x\cdot \pa_{\x}p_0(\x)}{|x||\pa_{\x}p_0(\x)|}$  at the characteristic surface $p_0(\x)=\l^2$. Hence, the radiation condition $\l^{-1}D_ru-u\in L^{2,-\frac{1}{2}+\e}$ means that $u$ is microlocally negligible away from the ougoing region $\{\frac{x\cdot \pa_{\x}p_0(\x)}{|x||\pa_{\x}p_0(\x)|}=1, p_0(\x)=\l^2, |x|>>1\}$, where we recall that $\Op(a)u=0$ implies $u$ is microlocally negligible in $\{a\neq 0\}$ by a standard elliptic estimate.

\subsection{Future works} 

There are some natural questions on the spectral properties for $P$. Here, we state some of them.

\noindent \textbf{Absense of embedded eigenvalues} It is natural to ask whether $P$ has an embedded $L^2$ eigenvalue. In this paper, we show this under Assumption \ref{timefct} and $\m>1$ for the negative spectrum, however, the general cases seem open.
For the case of Lorentzian scattering spaces, it is stated in \cite{V} that the operator $P$ has no embedded eigenvalue for the short-range case although its proof in detail is omitted and the short-range condition is assumed. By virtue of Corollary \ref{microsmooth} and Proposition \ref{abres} (or the argument in \cite{V}), for $\l\in\re \setminus \{0\}$, we have
\begin{align*}
(P-\l)u=0\,\, \text{and}\,\, u\in L^2(\re^{n+1})\Rightarrow u\in \mathcal{S}(\re^{n+1}).
\end{align*}
Thus, for proving absence of $L^2$ eigenvalues (with non-zero energy), it suffices to prove the absence of rapidly decreasing eigenfunction. Its proof leaves a future work.

\noindent \textbf{Limiting absorption principle at the zero energy} The problem on the existence of the outgoing resolvent at the zero energy (it corresponds to the massless case) is still open. Gell-Redman-Haber-Vasy \cite{GHV} deal with a related problem and constructed the massless Feynman propagator as an inverse between anisotropic b-Sobolev spaces when the perturbation $g-g_0$ is small enough (\cite[Theorem 3.6]{GHV}). For the exact Minkowski space with $n\geq 2$, it is proved in \cite[Proposition 2.8]{TT} that the outgoing/incoming resolvent $\jap{x}^{-1-0}(P_0\mp i0)^{-1}\jap{x}^{-1-0}\in B(L^2(\re^{n+1}))$ exist and $(P_0\mp i0)^{-1}$ map from $\jap{x}^{-1}L^2(\re^{n+1})$ to $\jap{x}L^2(\re^{n+1})$, where $P_0=\pa_t^2-\Delta_y$. The technique used there is based on the dispersive estimate and the Strichartz estimate for $e^{-itP_0}$, so their argument can not apply with our case directly.
 For a connection with the discrete Schr\"odinger operator, see \cite{TT}.

\subsection{Notation} We fix some notations. We denote $\jap{x}=(1+|x|^2)^{1/2}$ for $x\in \re^{n+1}$. The function space $\mathcal{S}(\re^{n+1})$ denotes the set of all rapidly decreasing functions on $\re^{n+1}$ and $\mathcal{S}'(\re^{n+1})$ denotes the set of all tempered distributions on $\re^{n+1}$. We use the weighted Sobolev space: $L^{2,l}(\re^{n+1})=\jap{x}^{-l}L^2(\re^{n+1})$, $H^k(\re^{n+1})=\jap{D}^{-k}L^2(\re^{n+1})$ and $H^{k,l}(\re^{n+1})=\jap{x}^{-l}\jap{D}^{-k}L^2(\re^{n+1})$ for $k,l\in \re$. For Banach spaces $X,Y$, $B(X,Y)$ denotes the set of all linear bounded operators form $X$ to $Y$. For a Banach space $X$, we denote the norm of $X$ by $\|\cdot\|_{X}$. If $X$ is a Hilbert space, we write the inner metric of $X$ by $(\cdot, \cdot)_{X}$, where $(\cdot, \cdot)_{X}$ is linear with respect to the right variable. We also denote $\|\cdot\|=\|\cdot\|_{L^2(\re^{n+1})}$ and $(\cdot,\cdot)=(\cdot, \cdot)_{L^2(\re^{n+1})}$.

\noindent\textbf{Acknowledgment.}  
This work is supported by JSPS Research Fellowship for Young Scientists, KAKENHI Grant Number 20J00221. The author would like to thank Christian G\'erard for suggesting the problem on the equivalence of the Feynman propagator and to Shu Nakamura for valuable discussions. Michal Wrochna explained the convention of the Feynman propagator in \cite{GW1}, \cite{GW2}, and in physics.

\section{Preliminary}

We define the Weyl quantization of a symbol $a$ by
\begin{align*}
\Op(a)u(x)=\frac{1}{(2\pi)^{n+1}}\int_{\re^{2n+2}}e^{i(x-x')\cdot \x}a(\frac{x+x'}{2},\x)u(x')dx'd\x.
\end{align*}
We note that $\Op(a)$ is formally self-adjoint for the standard $L^2$-space $L^2(\re^{n+1})$ if $a$ is real-valued.
Let $k,l\in \re$. For $a\in C^{\infty}(\re^{2n+2})$, we call $a\in S^{k,l}$ if
\begin{align*}
|\pa_{x}^{\a}\pa_{\x}^{\b}a(x,\x)|\leq C_{\a\b}\jap{x}^{l-|\a|}\jap{\x}^{k-|\b|}
\end{align*}
where $C_{\a\b}$ is independent of $(x,\x)\in \re^{2n+2}$. With the notation in \cite[Section 18.4]{Ho}, we can write $S^{k,l}=S(\jap{x}^l\jap{\x}^k ,g_{sc})$, where $g_{sc}=\jap{x}^{-2}dx^2+\jap{\x}^{-2}d\x^2$.
We define 
\begin{align*}
\{a,b\}:=H_ab:=\pa_{\x}a\cdot \pa_{x}b-\pa_{x}a\cdot \pa_{\x}b.
\end{align*}
We shall state the following properties of the quantization $\Op$ and and the symbol classes $S^{k,l}$.

\begin{lem}\label{prelimlem}
\item[$(i)$] Let $k,l\leq 0$ and $a\in S^{k,l}$. Then $\Op(a):\mathcal{S}(\re^n)\to \mathcal{S}(\re^n)$ can be uniquely extended to a bounded linear operator on $L^2(\re^n)$. 

\item[$(ii)$]
Let $k_j,l_j\in \re$ for $j=1,2$, $a\in S^{k_1,l_1}$ and $b\in S^{k_2,l_2}$. Then there exists $c\in S^{k_1+k_2, l_1+l_2}$ such that
\begin{align*}
\Op(c)=\Op(a)\Op(b).
\end{align*}
Moreover, we have
\begin{align*}
[\Op_h(a), i\Op_h(b)]=\Op(H_ab)+\Op S^{k_1+k_2-2, l_1+l_2-2},
\end{align*}
where we note
\begin{align*}
H_ab\in S^{k_1+k_2-1, l_1+l_2-1}.
\end{align*}

\item[$(iii)$]$($Sharp G\aa rding inequality$)$ Let $a\in S^{k,l}$ with $k,l\in \re$. Suppose $a(x,\x)\geq 0$ for $(x,\x)\in \re^{2n+2}$. Then there exists $C>0$ such that for $u\in \mathcal{S}(\re^{n+1})$
\begin{align*}
(u, \Op(a)u)\geq -C\|u\|_{H^{k-1,l-1}(\re^n)}^2.
\end{align*}
\end{lem}

For these proof, see \cite[\S18.4]{Ho}.

The proof of the following lemma is standard. 

\begin{lem}\label{convergence}
Let $k,\ell\in \re$. Assume $a_j\in S^{k,\ell}$ is a bounded sequence in $ S^{k,\ell}$ and $a_j\to 0$ in $S^{k+\d,\ell+\d}$ for some $\d>0$. Then, for each $s, t\in \re$ and $u\in H^{s,t}$
\[
\|\Op(a_j)u\|_{H^{s-\ell,t-k}}\to 0\quad\text{as } j\to \infty.
\]
In particular,
\end{lem}

\section{Limiting absorption principle}\label{LAPsec}

In this section, we always assume the null non-trapping condition (Assumption \ref{nulltrapp}). Moreover, we consider the operator $P_1=U^{-1}PU$ on the standard $L^2$-space $L^2(\re^{n+1})$, where $U:L^2(\re^{n+1})\to L^2(\re^{n+1}, \sqrt{|g|}dx)$ is a unitary operator defined by $Uu(x)=|g(x)|^{\frac{1}{4}}u(x)$.
For the simplicity of the notation, we denote $P$ by $P_1$ in this section. What is important is that $P$ has the form
\begin{align}\label{Pstd}
P=\Op(p)  +\Op S^{1,-\m}
\end{align}
and that $P$ is self-adjoint on $L^2(\re^{n+1})$.

To invoke the Mourre theory, we introduce the conjugate operator $A$:
\begin{align*}
a(x,\x)=\frac{x\cdot \tilde{\x}}{1+|\x|^2}\in S^{-1,1},\,\, A=\Op(a),\quad \tilde{\x}=\frac{1}{2}\pa_{\x}p_0(\x)=(-\x_1,\x_2,...,\x_{n+1}).
\end{align*}
By virtue of Nelson's commutator theorem, it follows that $A|_{C_c^{\infty}(\re^{n+1})}$ is essentially self-adjoint. We denote their unique self-adjoint extensions  by the same symbols $A$. We write the domain of the self-adjoint extension $P$ by $D(P)$, then we have
\begin{align}\label{domainP}
D(P)=\{u\in L^2(\re^{n+1})\mid Pu\in L^2(\re^{n+1})\}
\end{align}
by essential self-adjointness of $P|_{C_c^{\infty}(\re^{n+1})}$. Moreover, the essential self-adjointness of $P|_{C_c^{\infty}(\re^{n+1})}$ also implies 
\begin{align}\label{domaindense}
C_c^{\infty}(\re^{n+1})\,\, \text{is dense in}\,\,  D(P)\,\, \text{equipped with the graph norm of $P$}.
\end{align}

To prove Theorem \ref{lapwave}, we need some results on a pseudodifferential calculus of spectral cut-off functions for $P$ (Proposition \ref{pseudorem}) and the local compactness for $P$ (Corollary \ref{waveloccor}).
 We prove these results in the subsections later and deduce Theorem \ref{lapwave} here.

\begin{proof}[Proof of Theorem \ref{lapwave}]
We may assume $0<\m\leq 1$. We note $P\in C^2(A)$ since the linear operators $[P,iA]$ and $[[P,iA], iA]$ are bounded on $L^2(\re^{n+1})$.
Moreover, a standard argument based on the Hadamrad three line theorem gives
\begin{align*}
\jap{A}^s(P-i)^{-1}\jap{x}^{-s}\in B(L^2(\re^{n+1}))\quad \text{for}\quad 0\leq s\leq 1.
\end{align*}
Thus, we only have to prove Theorem \ref{lapwave} replacing $\jap{x}$ by $\jap{A}$.

Since $p-p_0\in  S^{2,-\m}$ and $a\in S^{-1,1}$, we have
\begin{align*}
\{p, a\}=\frac{4|\x|^2}{1+|\x|^2}+S^{0,-\m}.
\end{align*}
This implies 
\begin{align*}
[P, iA]= 4(I-\Delta)^{-1/2}(-\Delta)(I-\Delta)^{-1/2}+R,
\end{align*}
where $R\in \Op S^{0,-\m}$ (we note $0<\m\leq 1$). Let $J\Subset I$ be an open interval. Let $\f\in C_c^{\infty}(\re\setminus \{0\}; [0,1])$ which is supported in $I$ and is equal to $1$ on $J$. Moreover, take $\g_1\in C_c^{\infty}(\re^{n+1}; [0,1])$ which support is close to $0$ such that 
\begin{align*}
\supp\f\circ p\cap \supp\g_1=\emptyset.
\end{align*}
We observe 
\begin{align*}
(I-\Delta)^{-1/2}(-\Delta)(I-\Delta)^{-1/2}\geq c+\g(D),
\end{align*}
where
\begin{align*}
c=\inf_{\x\in \supp (1-\g_1)}\frac{|\x|^2}{1+|\x|^2}>0,\,\, \g(\x)=\frac{|\x|^2\g_1(\x)}{1+|\x|^2}-c\g_1(\x)\in C_c^{\infty}(\re^{n+1}).
\end{align*}
Thus we have
\begin{align}
\f(P) [P, iA]\f(P)\geq 4c\f(P)^2+4\f(P)\g(D)\f(P) +\f(P)R\f(P)
\end{align}
Proposition \ref{pseudorem} with a support propety $\supp\f\circ p\cap \supp\g=\emptyset$ implies that the second term of the right hand side is a compact operator on $L^2(\re^{n+1})$. Moreover, it follows that the third term $\f(P)R\f(P)$ is also compact by using the Helffer-Sj\"ostrand formula and the local compactness for $P$ (Corollary \ref{waveloccor}). From the Mourre theory \cite{Mo}, we obtain the desired results.

\end{proof}

\subsection{Pseudodifferential calculus of spectral cut-off functions}

In this subsection, we prove that $\f(D)\f(P)$ is a pseudodifferential operator plus a negligible term although $\f(P)$ itself cannot be written by such a form (even in the constant coefficient case $g=g_0$).
It is expected that $\g(D)\f(P)$ is actually a pseudodifferential operator of class $\Op S^{-\infty,0}$ (by using Beal's theorem). Here, we only show a weaker result which is needed for the Mourre estimate.

\begin{prop}\label{pseudorem}
Let $\f, \g\in C_c^{\infty}(\re)$. Then we can write
\begin{align*}
\g(D)\f(P)=\Op(\g\cdot \f\circ p)+K,
\end{align*}
where $K$ is a compact operator on $L^2(\re^{n+1})$.
\end{prop}
\begin{proof}
First, we construct a parametrix of $\g(D)(P-z)^{-1}$ for $\Im z\neq 0$. We note that for any integer $N\geq 0$, we have 
\begin{align*}
|\pa_{x}^{\a}\pa_{\x}^{\b}\left(\g(\x)(p(x,\x)-z)^{-1}\right)|\leq C_{N\a\b}|\Im z|^{-|\a|-|\b|-1}\jap{x}^{-|\a|}\jap{\x}^{-N}
\end{align*}
with a constant $C_{N\a\b}>0$ independent of $z$ and $(x,\x)\in \re^{2n+2}$. This implies
\begin{align*}
\g(\x)=\frac{\g}{p-z}\# (p-z)(x,\x)+r_z(x,\x),
\end{align*}
where $r_0\in S^{-\infty,-1}$ satisfying
\begin{align}\label{funrem}
|\pa_{x}^{\a}\pa_{\x}^{\b}r_z(x,\x)|\leq C_{N\a\b}|\Im z|^{-N_{\a\b}}\jap{x}^{-1-|\a|}\jap{\x}^{-N}
\end{align}
with a constant $C_{N\a\b}>0$ and $N_{\a\b}\geq 0$ independent of $z$ and $(x,\x)\in \re^{2n+2}$. Weyl quantizing this equation and multiplying $(P-z)^{-1}$ from left, we have
\begin{align*}
\g(D)(P-z)^{-1}=\Op(\frac{\g}{p-z})(I+(P-\Op(p))(P-z)^{-1}) +\Op(r_z)(P-z)^{-1}
\end{align*}
as a bounded operator on $L^2(\re^{n+1})$. 

Now we denote the almost analytic extension \cite[Theorem 3.6]{Z} of $\f$ by $\tilde{\f}$. By the Helffer-Sj\"ostrand formula \cite[Theorem 14.8]{Z}, we have
\begin{align*}
\g(D)\f(P)=&\Op(\g\cdot \f\circ p)+\frac{1}{\pi i}\int_{\mathbb{C}}\bar{\pa_{z}}\tilde{\f}(z)\Op(\frac{\g}{p-z})(P-\Op(p))(P-z)^{-1}dz\\
&+\frac{1}{\pi i}\int_{\mathbb{C}}\bar{\pa_{z}}\tilde{\f}(z)\Op(r_z)(P-z)^{-1}dz.
\end{align*}
The estimates $(\ref{funrem})$, $\|(P-z)^{-1}\|_{B(L^2)}\leq |\Im z|^{-1}$ and $\bar{\pa_{z}}\tilde{\f}(z)=O(|\Im z|^{\infty})$ as $\Im z\to 0$ imply that the third term of the right hand side is a bounded operator from $L^2(\re^{n+1})$ to $H^{1,1}$. Moreover, since the operator $\bar{\pa_{z}}\tilde{\f}(z)\Op(\frac{\g}{p-z})(P-\Op(p))$ is compact and its norm is bounded by $C|\Im z|$ with $C>0$, the second term is also compact.
 Since the natural injection $H^{1,1}\hookrightarrow L^2(\re^{n+1})$ is compact, we obtain the desired result.

\end{proof}

\subsection{Local compactness}

In this subsection, we prove the local compactness for $P$. The main result of this subsection is the following proposition.
\begin{prop}\label{waveloc}
There exists $C>0$ such that
\begin{align}\label{waveloces}
\|u\|_{H^{\frac{1}{2},-\frac{1+\d}{2}}}\leq C\|Pu\|_{L^2}+C\|u\|_{L^2}
\end{align}
for $u\in D(P)=\{u\in L^2(\re^{1+n})\mid Pu\in L^2(\re^{1+n})\}$. In particular, we have a continuous inclusion
\begin{align*}
D(P)\hookrightarrow H^{\frac{1}{2},-\frac{1+\d}{2}}.
\end{align*}
where we regard $D(P)$ as a Banach space equipped with the graph norm of $P$.
\end{prop}

\begin{rem}
This proposition is not a direct consequence of the radial estimate. In fact, for applying the radial source estimate, we need an addition decay for $Pu$.
\end{rem}

\begin{cor}\label{waveloccor}
Let $W\in C(\re^{1+n})$ satisfying $|W(x)|\to 0$ as $|x|\to 0$. Then it follows that $W(P-i)^{-1}$ is a compact operator on $L^2(\re^{1+n})$. In particular, if $A\in \Op S^{0,-\e}$ with $\e>0$, then $A(P-i)^{-1}$ is also compact.
\end{cor}

\begin{proof}[Proof of Corollary \ref{waveloccor}]
Let $W \in C(\re^{1+n})$ satisfying $|W(x)|\to 0$ as $|x|\to 0$. Then there exists $W_k\in C_c^{\infty}(\re^{1+n})$ such that $\|W_k-W\|_{L^{\infty}(\re^{1+n})}\to 0$ as $k\to \infty$. Since the multiplication operator $W_k$ is continuous from $H^{\frac{1}{2},-\frac{1+\d}{2}}$ to $H^{\frac{1}{2},1}$ and the natural inclusion $H^{\frac{1}{2},1}\hookrightarrow L^2(\re^{1+n})$ is compact, then Proposition \ref{waveloc} implies that $W_k(P-i)^{-1}$ is also compact in $L^2(\re^{1+n})$. Since a limit of compact operators is also compact, then we conclude the compactness of $W(P-i)^{-1}$.
\end{proof}

In the following, we show Proposition \ref{waveloc}.
Now Proposition \ref{waveloc} follows from existence of the following escape function.

\begin{lem}[Escape function under null non-trapping condition]\label{escapenullnon}
Let $0<2\d<\m$. There exist $\l_0>0$, $C_1>0$ and $a\in S^{0,0}$ such that
\begin{align*}
H_pa(x,\x)\geq C_1\jap{x}^{-1-\d}\jap{\x}-r(x,\x),
\end{align*}
where $r\in S^{1,-1}$ satisfies
\begin{align*}
\supp r\subset \{(x,\x)\in \re^{2n+2}\mid |\x|\leq 2\}\cup \{(x,\x)\in \re^{2n+2}\mid |p(x,\x)|\geq \frac{\l_0}{2}|\x|^2\}.
\end{align*}

\end{lem}

\begin{proof}[Proof of Proposition \ref{waveloc} assuming Lemma \ref{escapenullnon}]
We may assume $0<2\d<\m$. By $(\ref{domaindense})$, it suffices to prove $(\ref{waveloces})$ for $u\in \mathcal{S}(\re^{1+n})$.
By using the sharp G\aa rding inequality and using $A\in \Op S^{0,0}$, for $u\in \mathcal{S}(\re^{1+n})$, we have
\begin{align}\label{onchar}
\|u\|_{H^{\frac{1}{2},-\frac{1+\d}{2}}}^2\leq C\|Pu\|_{L^2}^2+C\|u\|_{L^2}^2+C|(u,\Op(r)u)_{L^2}|
\end{align}
with a constant $C>0$. Now we write
\begin{align*}
&r=r_1+r_2,\,\, r_1\in S^{-\infty,-1},\,\, r_2\in S^{1,-1},\\
&\supp r_1\subset \{(x,\x)\in \re^{2n+2}\mid |\x|\leq 4\},\,\, \supp r_2\subset \{(x,\x)\in \re^{2n+2}\mid |\x|\geq 3,\, |p(x,\x)|\geq \frac{\l_0}{4}|\x|^2 \}.
\end{align*}
By the standard elliptic parametrix construction, we have
\begin{align}\label{awaychar}
|(u,\Op(r_1)u)_{L^2}|\leq C\|u\|_{L^2}^2,\,\, |(u,\Op(r_2)u)_{L^2}|\leq C\|Pu\|_{L^2}^2+ C\|u\|_{L^2}^2
\end{align}
for $u\in \mathcal{S}(\re^{1+n})$. Combining $(\ref{onchar})$ with $(\ref{awaychar})$, we obtain $(\ref{waveloces})$ for $u\in \mathcal{S}(\re^{1+n})$.
\end{proof}

To prove Lemma \ref{escapenullnon}, we need some preliminary lemmas.

\begin{lem}[Convexity at infinity 1]\label{conv1}
There exists $R_0>0$ such that for $(x,\x)\in T^*\re^{1+n}$ with $|x|\geq R_0$, we have
\begin{align*}
H_p^2|x|^2\geq C|\x|^2.
\end{align*}
\end{lem}
\begin{proof}
This lemma follows from an easy calculation.
\end{proof}

\begin{lem}[Convexity at infinity 2]\label{conv2}
Let $R\geq R_0$, where $R_0$ be same as that of Lemma \ref{conv1}. If $t_0<t_1$ and $(x,\x)\in T^*\re^{1+n}$ satisfy 
\begin{align*}
|z(t_j,x,\x)|\leq R.
\end{align*}
for $j=1,2$. Then for $t\in [t_1,t_2]$, we have
\begin{align*}
|z(t,x,\x)|\leq R.
\end{align*}
\end{lem}

\begin{proof}
This lemma immediately follows from Lemma \ref{conv1}, where we note
\begin{align*}
\frac{d^2}{dt^2}|z(t,x,\x)|^2=(H_p^2|x|^2)|_{x=z(t,x,\x),\, \x=\z(t,x,\x)}.
\end{align*}

\end{proof}

We denote
\begin{align*}
D_R=\{x\in \re^{1+n}\mid |x|\leq R\},\,\, S^*D_R=\{(x,\x)\in T^*\re^{1+n}\mid |x|\leq R,\,\, |\x|=1\}.
\end{align*}

\begin{lem}[Stability of non-trapping orbit]\label{stab}
We assume that for $(x,\x)\in p^{-1}(\{0\})\setminus\{\x=0\}$, we have $|z(t,x,\x)|\to \infty$ as $|t|\to \infty$. Let $R\geq R_0$, where $R_0$ be same as that of Lemma \ref{conv1}. Then there exists $\l_0>0$ and $T> 1$ such that we have
\begin{align*}
|z(t,x,\x)|>R\,\, \text{for}\,\, |t|\geq T,\,\, (x,\x)\in  p^{-1}([-\l_0,\l_0]) \cap S^*D_R(0).
\end{align*}

\end{lem}

\begin{proof}
By the assumption and Lemma \ref{conv2}, for any $(x,\x)\in p^{-1}(\{0\})\cap S^*D_R$ there exist $T(x,\x)>0$ and a neighborhood $U(x,\x)\subset T^*\re^{1+n}$ of $(x,\x)$ such that
\begin{align}\label{convnontr}
|z(t,y,\y)|> R \,\, \text{for}\,\,  |t|\geq T(x,\x),\,\, (y,\y)\in U(x,\x).
\end{align}
We prove this for $t\geq 0$. By the non-trapping assumption there exists $T(x,\x)$ such that 
\begin{align*}
|z(T(x,\x),x,\x)|> R+1.
\end{align*}
Since the set $\{(y,\y)\in T^*\re^{1+n}\setminus \{\y=0\} \mid(|z(T(x,\x),y,\y)|>R+1\}$ is open, there exists a neighborhood $U(x,\x)\subset T^*D_{R+1}$ of $(x,\x)$ such that 
\begin{align*}
|z(T_0(x,\x),y,\y)|> R+1\,\, \text{for}\,\, (y,\y)\in U(x,\x).
\end{align*}
Thus we have
\begin{align*}
|z(t,y,\y)|> R+1\,\, \text{for}\,\, t\geq T(x,\x), \,\, (y,\y)\in U(x,\x).
\end{align*}
This proves $(\ref{convnontr})$.

Since the set $p^{-1}(\{0\})\cap S^*D_R$ is compact, there are finite many point $\{(x_j,\x_j)\}_{j=1}^N\subset p^{-1}(\{0\})\cap S^*D_R$ such that
\begin{align}\label{convinc}
p^{-1}(\{0\})\cap S^*D_R\subset \bigcup_{j=1}^NU(x_j,\x_j)=:U.
\end{align}
We set $T=\max_{1\leq j\leq N}T(x_j,\x_j)$. Then we have
\begin{align*}
|z(t,x,\x)|>R\,\, \text{for}\,\, |t|\geq T,\,\, (x,\x)\in  U.
\end{align*}
Thus it suffices to prove that there is $\l_0>0$ such that
\begin{align*}
p^{-1}([-\l_0,\l_0])\cap S^*D_R\subset U.
\end{align*}
To prove this, we suppose that for any $k\in \mathbb{N}$, there exists $\rho_k\in p^{-1}([-1/k,1/k])\cap S^*D_R$ such that $\rho_k\in  U^c$. Since $p^{-1}([-1,1])\cap S^*D_R$ is compact, there exist a subsequence $\rho_{k_l}$ and $\rho\in p^{-1}([-1,1])\cap S^*D_R$ such that $\rho_{k_l}\to \rho$. However, this concludes $\rho\in p^{-1}(\{0\})\cap S^*D_R\cap U^c$ since $U^c$ is closed. This contradicts to $(\ref{convinc})$.
\end{proof}

\begin{lem}[Escape function on a compact set]\label{escapecompact}
Let $\l_0>0$ be as in Lemma \ref{stab}.
Let $\chi\in C_c^{\infty}(\re;\re)$ and set
\begin{align*}
a_0(x,\x)=\int_0^{\infty}\chi(z(t,x,\x))|\z(t,x,\x)|dt.
\end{align*}
Let $R\geq R_0$, where $R_0$ be same as that of Lemma \ref{conv1}.
Then $a_0$ is well defined smooth function on the set
\begin{align*}
C_{\l_0}:=\{(x,\x)\in T^*\re^{1+n}\mid |x|\leq R,\, |\x|\geq 1,\,\, |p(\x)|<\l_0|\x|^2\}
\end{align*}
and $a_0$ satisfies
\begin{align}\label{esces}
|\pa_{x}^{\a}\pa_{\x}^{\b}a_0(x,\x)|\leq C_{\a\b}\jap{\x}^{-|\b|}\,\,\text{for}\,\, (x,\x)\in C_{\l_0}.
\end{align}

\end{lem}

\begin{proof}
We take $T>0$ same as that of Lemma \ref{stab}.
We note $(z(t,x,\x),\z(t,x,\x))=(z(|\x|t,x,\frac{\x}{|\x|}), |\x|\z(|\x|t,x,\frac{\x}{|\x|}))$ and 
\begin{align*}
a_0(x,\x)=&|\x|\int_0^{\infty}\chi(z(|\x|t,x,\frac{\x}{|\x|}))|\z(|\x|t,x,\frac{\x}{|\x|})|dt\\
=&\int_0^{T}\chi(z(t,x,\frac{\x}{|\x|}))|\z(t,x,\frac{\x}{|\x|})|dt
\end{align*}
for $(x,\x)\in C_{\l_0}$. Thus it follows that $a_0$ is a well-defined smooth function. We note
\begin{align*}
|\pa_{x}^{\a}\pa_{\x}k(t,x,\x)|\leq C_{\a\b},\,\, k\in \{z,\z\}
\end{align*}
uniformly in $|t|\leq T$, $|x|\leq R$ and $|\x|=1$, and
\begin{align*}
|\pa_{\x}^{\b}\frac{\x}{|\x|}|\leq C_{\b}\jap{\x}^{-\b},\,\, |\x|\geq 1.
\end{align*}
These inequalities give $(\ref{esces})$.

\end{proof}

Now we prove Lemma \ref{escapenullnon}.

\begin{proof}[Proof of Lemma \ref{escapenullnon}]
Let $\l_0>0$ be as in Lemma \ref{stab}. 
We fix some notations which works only in this subsection. For $\l_0>0$, let $\g_{\l_0}\in C_c^{\infty}(\re;[0,1])$ such that
\begin{align}\label{Kleincut}
\supp \tilde{\g}_{\l_0}\subset (-\l_0,\l_0),\,\, \tilde{\g}_{\l_0}(t)=1\,\, \text{for}\,\, t\in (-\frac{\l_0}{2}, \frac{\l_0}{2}).
\end{align}
Moreover, we take $\chi \in C_c^{\infty}(\re, [0,1])$ such that
\begin{align*}
\chi(t) =\begin{cases}
1,\quad t\leq 1,\\
0, \quad t\geq 2,
\end{cases}\quad \chi'(t)\leq 0
\end{align*}
and $\chi(t)>0$ for $1<t<2$.
We set $\chi_R(t)=\chi(t/R)$ and $\bar{\chi}_R(t)=1-\chi_R(t)$ for $R>0$.
\begin{align*}
\g_{\l_0}(x,\x)=\bar{\chi}(|\x|)\tilde{\g}_{\l_0}(\frac{p(x,\x)}{|\x|^2}).
\end{align*}
We define
\begin{align*}
C_{\l_0}(P)=\{(x,\x)\in \re^{2n+2}\mid |\x|\geq 1,\,\, |p(\x)|< \l_0 |\x|^2\},
\end{align*}
then we note $\supp \g_{\l_0}\subset C_{\l_0}(P)$ and $\g_{\l_0}\in S^{0.0}$.

It suffices to construct a smooth real-valued function $q$ such that for $(x,\x)\in C_{\l_0}(P)$, we have
\begin{align}\label{nullrein}
|\pa_{x}^{\a}\pa_{\x}^{\b}q(x,\x)|\leq C_{\a\b}\jap{x}^{-|\a|}\jap{\x}^{-|\b|},\,\, H_pq(x,\x)\geq C\jap{x}^{-1-\d}\jap{\x}.
\end{align}
In fact, setting $a(x,\x)=\g_{\l_0}(x,\x)^2q(x,\x)$ and $r(x,\x)=C\jap{x}^{-1-\d}\jap{\x}(1-\g_{\l_0}(x,\x)^2)+qH_p\g_{\l_0}^2(x,\x)$, then $a$ and $r$ satisfy the desired property.

We introduce a function
\begin{align*}
\b(x,\x)=\frac{x\cdot \pa_{\x}p(x,\x)}{|x||\pa_{\x}p(x,\x)|}.
\end{align*}
For $M, L>0$ which are large enough and determined later, we set
\begin{align*}
q_1(x,\x)=&\b(x,\x)\int_1^{2|x|}\frac{1}{s^{1+\d}}ds\bar{\chi}_{M}(|x|),\,\,
q_2(x,\x)=\chi_{2M}(|x|)\int_{\infty}^0\frac{\chi_{M}(z(t,x,\x))}{\jap{z(t,x,\x)}^{1+\d}}|\z(t,x,\x)|dt,\\
q(x,\x)=&Lq_1(x,\x)+q_2(x,\x),
\end{align*}
where $q_2$ is well-defined for $(x,\x)\in C_{\l_0}(P)$ by Lemma \ref{escapecompact}.
We claim that for $M\geq 1$ large enough,
\begin{align}\label{locineq}
H_{p}(\b\int_{1}^{2|x|}\frac{1}{s^{1+\d}}ds)\geq& C_2\jap{x}^{-1-\d}\jap{\x}\,\,\text{for}\,\, |x|\geq M,\,\, |\x|\geq 1,
\end{align}
with a constant $C_2>0$. In fact, we have
\begin{align*}
H_{p}(\b\int_{1}^{2|x|}\frac{1}{s^{1+\d}}ds)=&\frac{|\pa_{\x}p(x,\x)|}{|x|}(1-\b^2)\int_{1}^{2|x|}\frac{1}{s^{1+\d}}ds+\frac{\b^2|\pa_{\x}p(x,\x)|}{2^{\d}|x|^{1+\d}}+O(\jap{x}^{-1-\m}|\x|).
\end{align*}
This calculation gives $(\ref{locineq})$. 
Using the inequality $\b H_{p}\overline{\chi}_{M}(|x|)\geq 0$, we have
\begin{align*}
H_pq_1(x,\x)\geq&  C_2\jap{x}^{-1-\d}\jap{\x}\bar{\chi}_{M}(|x|)\,\,\text{for}\,\, |\x|\geq 1.
\end{align*}
Moreover, there exist $C_4,C_5,C_6>0$ such that for $(x,\x)\in C_{\l_0}(P)$ and $M\geq 1$, we have
\begin{align*}
H_pq_2(x,\x)=&\chi_{M}(|x|)\jap{x}^{1+\d}|\x|+H_p(\chi_{2M}(|x|))\int_{\infty}^0\frac{\chi_{M}(z(t,x,\x))}{\jap{z(t,x,\x)}^{1+\d}}|\z(t,x,\x)|dt\\
\geq& C_4\jap{x}^{-1-\d}\jap{\x}\chi_{M}(|x|)-C_5\jap{x}^{-1}\jap{\x}\bar{\chi}_{M}(|x|)\chi_{4M}(|x|)\\
\geq&C_4\jap{x}^{-1-\d}\jap{\x}\chi_{M}(|x|)-C_6\jap{x}^{-1-\d}\jap{\x}\bar{\chi}_{M}(|x|).
\end{align*}
Thus we have
\begin{align*}
H_pq(x,\x)\geq& (LC_2-C_6)\jap{x}^{-1-\d}\jap{\x}\bar{\chi}_{M}(|x|)+C_4\jap{x}^{-1-\d}\jap{\x}\chi_{M}(|x|)
\end{align*}
for $(x,\x)\in C_{\l_0}(P)$.
Taking $L>0$ such that $L>2C_6/C_2$, we obtain $(\ref{nullrein})$. 

\end{proof}

\begin{rem}
When we assume the globally non-trapping assumption, a more stronger equality holds: There exist $C_1, C_2>0$ and $a\in S^{0,0}$ such that
\begin{align*}
H_pa(x,\x)\geq C_1\jap{x}^{-1-\d}\jap{\x}-C_2.
\end{align*}
In fact, in the proof above, we only have to replace $\g_{\l_0}$ by $\bar{\chi}(|\x|)$. 
\end{rem}


\section{Feynman propagator in G\'erard-Wrochna's theory}

Throughout of this section, we assume $\m>1$, Assumptions \ref{nulltrapp} and \ref{timefct}. Let $m_0>0$.

\subsection{Reviews on G\'erard-Wrochna's theory}\label{GWsusec}

In this subsection, we review some notation and results in \cite{GW1} and \cite{GW2}. In \cite{GW1} and \cite{GW2}, the authors sometimes use identification of the notation which is slightly confusing. In this paper, we use the notation to distinguish them for the sake of ease of understanding at the cost of its simplicity.

\subsection*{Global hyperbolicity}

By \cite[Lemma 4.4]{GW1}, after replacing $\tilde{t}$ and $t$ by $\tilde{t}-c$ and $t-c$ with a large constant $c>>1$, we find a diffeomorphism $\chi:\re^{n+1}\to \re^{n+1}$ such that
\begin{align*}
\chi^*g=-\hat{c}(\hat{t},\hat{y})^2d\hat{t}^2+\sum_{i,j=1}^n\hat{h}_{ij}(\hat{t},\hat{y})d\hat{y}^id\hat{y}^j,\quad \tilde{t}(\chi(\hat{t},\hat{y}))=\hat{t},
\end{align*}
where $\hat{c}$ is a positive smooth function and $\hat{h}=\sum_{i,j=1}^n\hat{h}_{ij}(\hat{t},\hat{y})dy^idy^j$ is a Riemannian metric on $\re^n$ smoothly depending on $t\in \re$. Moreover, if we denote the flat Riemannian metric on $\re^n$ by $h_{free}$, there exist diffeomorphisms $y_{out/in}$ on $\re^n$ such that
\begin{align}\label{asyp1}
y_{out/in}-id\in S^{1-\m}(\re^n),\,\, \hat{h}-(y_{out/in})^*h_{free}\in S^{-\m}(\re^{\pm}\times\re^{n}),\,\,  \hat{c}-1       \in S^{-\m}(\re^{n+1}).
\end{align}
If we denote $\chi_{\hat{t}}(\hat{y})=\pi_{\yh}\chi(\th,\yh)$(where $\pi_{\yh}$ is the projection into the $\yh$-variable), then
\begin{align}
&\chi(\hat{t},\hat{y})=(\hat{t}, \chi_{\th}(\yh))\quad \text{outside a compact set},\nonumber\\
&\chi_{\th}(\yh)-y_{out/in}(\yh)\in S^{1-\m}(\re^{\pm}\times \re^n).\label{asyp2}
\end{align}
Furthermore, as is (implicitly) shown in the proof of \cite[Lemma 4.4]{GW1}, 
\begin{align*}
\chi(\hat{t},\hat{y})=(\hat{t},\hat{y})+S^{1-\m}(\re^{n+1})
\end{align*}
holds. The above relations imply the following lemma.

\begin{lem}\label{deffeoprop}
For $k,l\in \re$, the linear operators $\chi^*, (\chi^{-1})^*:H^{k,l}(\re^{n+1}) \to H^{k,l}(\re^{n+1})$ are bounded. In particular, $\chi^*,(\chi^{-1})^*$ preserve the Schwartz space $\mathcal{S}(\re^{n+1})$. Moreover, the linear operators $\chi_{\hat{t}}^*, (\chi_{\hat{t}}^{-1})^*:H^{k,l}(\re^{n}) \to H^{k,l}(\re^{n})$ are bounded and these operator norms are uniformly bounded in $\hat{t}$ if $|\hat{t}|$ is large enough.
\end{lem}

\subsection*{Conformal transformation}

Define 
\begin{align}\label{deftildeP}
\tilde{P}:=\chi^* P(\chi^{-1})^*=-\Box_{\chi^*g}+\chi^*V,  \quad \tilde{P}:= \hat{c}^{\frac{n+1}{2}+1}\hat{P}\hat{c}^{1-\frac{n+1}{2}},
\end{align}
where we note that the definition of $\tilde{P}$ is slightly different from that in \cite[before (4.6)]{GW1}. To explain why we define $\tilde{P}$ as in $(\ref{deftildeP})$, we recall a useful formula in the conformal geometry.
In general, if we denote the scalar curvature of a metric $g_1$ by $R_{g_1}$, it is well-known that the conformal Laplacian $-\Box_{g_1}+\frac{n-1}{4n}R_{g_1}$ (where the dimension of the manifold is $n+1$) is changed under a conformal diffeomorphism  as follows:
\begin{align*}
e^{-\frac{n+3}{2}f}(-\Box_{g_1}+\frac{n-1}{4n}R_{g_1})e^{\frac{n-1}{2}f}=-\Box_{e^{2f}g_1}+\frac{n-1}{4n}R_{e^{2f}g_1}.
\end{align*}
Setting $g_1=\chi^*g$ and $e^{f}=\hat{c}^{-1}$, we have
\begin{align*}
\tilde{P}=&-\Box_{\hat{c}^{-2}\chi^*g}+\frac{n-1}{4n}(R_{\hat{c}^{-2}\chi^*g}- \hat{c}^2R_{\chi^*g})+\hat{c}^2(\chi^*V)\\
=&\pa_{\hat{t}}^2+r\pa_{\hat{t}}-\Delta_{\hat{c}^{-2}\hat{h}}+\tilde{V}
\end{align*}
where $r=|\hat{c}^{-2}\hat{h}|^{-\frac{1}{2}}\pa_{\hat{t}}|\hat{c}^{-2}\hat{h}|$, $\tilde{V}=\frac{n-1}{4n}(R_{\hat{c}^{-2}\chi^*g}- \hat{c}^2R_{\chi^*g})+\hat{c}^2(\chi^*V)\in S^{-\m}(\re^{n+1})$. Then $\tilde{P}+m_0^2$ is of the form as in \cite[(3.1)]{GW1}. We note that the operator \cite[before (4.6)]{GW1} is also of the form as in \cite[(3.1)]{GW1}, although its calculation is more complicated. See also \cite[\S 5.2]{GW0}

\subsection*{Cauchy hypersurface}

We write the Cauchy hypersurface (in the $ty$-plane) for $P$ by $\Sigma_t:=\tilde{t}^{-1}(\{t\})=\chi(\{t\}\times \re^n)$. Then it follows from $\tilde{t}-t\in C_c^{\infty}(\re^{n+1})$ that 
\begin{align}\label{Caut}
\Sigma_t=\{t\}\times \re^n_y\quad \text{if $|t|$ is large enough}.
\end{align}
Moreover,  replacing $\tilde{t}$ and $t$ by $\tilde{t}-c$ and $t-c$ with a large constant $c>>1$ again, we may assume
\begin{align}\label{Cau0}
\Sigma_0=\{0\}\times \re^n_y.
\end{align}

\subsection*{Cauchy evolution}

We denote
\begin{align*}
\rho_{0,t}u(y):=\begin{pmatrix}
u(t,y)\\
D_{t}u(t,y)
\end{pmatrix},\,\,
\hat{\rho}_{\hat{t}}\hat{u}:=\begin{pmatrix}
\hat{u}(\hat{t},\hat{y})\\
\hat{c}^{-1}(\hat{t},\hat{y}) D_{\hat{t}}\hat{u}(\hat{t},\hat{y})
\end{pmatrix},\,\,
\rho_{\tilde{t}}u:=\begin{pmatrix}
u|_{\Sigma_{\tilde{t}}}\\
i^{-1}n\cdot \nabla u|_{\Sigma_{\tilde{t}}}
\end{pmatrix},
\end{align*}
where $n$ is the future directed unit normal to $\Sigma_t$. Setting $(\rho u)(t)=\rho_tu$ and $(\hat{\rho} \hat{u})(\hat{t})=\hat{\rho}_{\hat{t}}\hat{u}$, we have
\begin{align*}
\hat{\rho}\hat{u}=\chi^*(\rho u).
\end{align*}
For $m\in \re$ and $\frac{1}{2}<\c<\min(\frac{1}{2}+\m,1)$, we define
\begin{align*}
&\mathcal{E}^m=H^{m+1}(\re^n)\oplus H^m(\re^n),\quad \hat{\mathcal{Y}}^m=\jap{\hat{t}}^{-\c}L^2(\re; H^m(\re^n)),\quad \mathcal{Y}^m=(\chi^{-1})^*(\hat{\mathcal{Y}}^m),  \\
&\hat{\mathcal{X}}^m=\{\hat{u}\in  C(\re; H^{m+1}(\re^n))\cap C^1(\re; H^m(\re^n))   \mid  (\hat{P}+m_0^2)\hat{u}\in \hat{\mathcal{Y}}^m\},\,\, \mathcal{X}^m=(\chi^{-1})^*(\hat{\mathcal{X}}^m).
\end{align*}
By the existence and uniqueness of the solution to the Cauchy problem (\cite[Lemma 3.5]{GW1}), the function space $\hat{\mathcal{X}}^m$ is the Hilbert space equipped with the norm
\begin{align*}
\|\hat{u}\|_{m}^2=\|\hat{\rho}_0\hat{u}\|_{\mathcal{E}^m}^2+\|(\hat{P}+m_0^2)\hat{u}\|_{\hat{\mathcal{Y}}^m}^2.
\end{align*}
For $t,s\in \re$, let us denote by $\mathcal{U}_{free}(t,s)$ the Cauchy evolution propagator of $P_0+m_0^2$:
\begin{align*}
\mathcal{U}_{free}(t,s)=&e^{itH_{free}}=
\begin{pmatrix}
\cos t A_{free}&i A_{free}^{-1}\sin tA_{free}\\
i A_{free}\sin tA_{free}&\cos t A_{free}
\end{pmatrix}:\mathcal{E}^m\to \mathcal{E}^m,\\
H_{free}=&\begin{pmatrix}
0&1\\
A_{free}^2&0
\end{pmatrix},\quad A_{free}=\sqrt{-\Delta_y+m_0^2}.
\end{align*}
We also denote the Cauchy evolution for $P+m_0^2$ and $\tilde{P}+m_0^2\hat{c}^2$ by $\mathcal{U}(t,s)$ and $\tilde{\mathcal{U}}(t,s)$ respectively.
We define
\begin{align*}
\quad R(t)=\hat{c}^{1-\frac{n+1}{2}}\begin{pmatrix}
1&0\\
i\frac{n-1}{2}\hat{c}^{-1}\pa_{\hat{t}}\log\hat{c}&\hat{c}^{-1}
\end{pmatrix}.
\end{align*}
Setting $\tilde{u}=\hat{c}^{\frac{n+1}{2}-1}\hat{u}$, we have $(\tilde{P}+m_0^2\hat{c}^2)\tilde{u}=0\Leftrightarrow (\hat{P}+m_0^2)\hat{u}=0$. If we define $\iota_t:\re^n\to \{t\}\times \re^n$ and $\a_t:\re^n\to \Sigma_t$ by $\iota_t(y)=(t,y)$ and $\a_t(\hat{y})=\chi(t,\hat{y})$ and  respectively, then we have $U(t,s)=(\a_t^{-1})^* R(t)\tilde{U}(t,s)R(s)^{-1}\a_s^*$ and
\begin{align*}
\mathcal{U}(t,s):=\iota_t^*U(t,s)(\iota_s^{-1})^*=(\chi_t^{-1})^* R(t)\tilde{U}(t,s)R(s)^{-1}\chi_s^*\in B(\mathcal{E}^m),
\end{align*}
where the latter operator is well-defined if $\Sigma_{\ast}=\{\ast\}\times \re^n$ with $\ast\in \{s,t\}$ (which is true for $\ast=0$ or for large $|\ast|$ by virtue of $(\ref{Caut})$ and $(\ref{Cau0})$). By the Cock-method, the inverse wave operators
\begin{align*}
W_{\pm}^{-1}:=\lim_{t\to\pm\infty}\mathcal{U}_{free}(0,t)  \mathcal{U}(t,0) \in B(\mathcal{E}^m)
\end{align*}
exist.

\subsection*{Boundary conditions}

Set
\begin{align*}
c_{free}^{\pm,vac}=\frac{1}{2}\begin{pmatrix}
1&\pm A_{free}^{-1}\\
\pm A_{free}&1
\end{pmatrix},\quad \pi^+=\begin{pmatrix}
1&0\\
0&0
\end{pmatrix},\quad \pi^-=\begin{pmatrix}
0&0\\
0&1
\end{pmatrix}.
\end{align*}
We define $\rho_{F}$ and $\rho_{\overline{F}}$ by
\begin{align*}
\rho_{F}:=&s-\lim_{t_{\pm}\to \pm \infty}(c_{free}^{+,vac}\mathcal{U}_{free}(0,t_+)\rho_{t_+}+c_{free}^{-,vac}\mathcal{U}_{free}(0,t_-)\rho_{t_-} ),\\
\rho_{\overline{F}}:=&s-\lim_{t_{\pm}\to \pm \infty}(c_{free}^{+,vac}\mathcal{U}_{free}(0,t_-)\rho_{t_-}+c_{free}^{-,vac}\mathcal{U}_{free}(0,t_+)\rho_{t_+} ).
\end{align*}
We define
\begin{align*}
\mathcal{X}^m_{F/\overline{F}}:=\{u\in \mathcal{X}^m\mid \rho_{\overline{F}/F}u=0\}.
\end{align*}

\subsection{Absence of embedded eigenvalues}

In this subsection, we prove Theorem \ref{Feyeq} $(i)$. This part essentially follows from the result in \cite{GW2} and the radial estimate.

\begin{proof}[Proof of Theorem \ref{Feyeq} $(i)$.]
Suppose $u\in L^2(\re^{n+1})$ satisfies $(P+m_0^2)u=0$. By Corollary \ref{microsmooth} and Proposition \ref{abres}, we have $u\in \mathcal{S}(\re^{n+1})$. By Lemma \ref{deffeoprop}, we have $\mathcal{S}(\re^{n+1})\subset \mathcal{X}^m_{F}$ for all $m\in \re$. Since the map $P+m_0^2:\mathcal{X}^m_F\to \mathcal{Y}^m$ is injective (\cite[Theorem 1.1]{GW2}), we conclude $u=0$.

\end{proof}

\subsection{Asymptotic expansion of the wave}

The following formula is essentially used in the proof of \cite[Lemma 2.2]{GW2} for the diagonalized operator $P^{ad}$.

\begin{prop}\label{waveexprop}
Suppose $u\in \mathcal{X}^m$. Then there exist $g_{\pm,\pm}\in H^m(\re^n)$ such that
\begin{align}\label{waveexp}
u(t)=e^{it\sqrt{-\Delta+m_0^2}}g_{\pm,+}+e^{-it\sqrt{-\Delta+m_0^2}}g_{\pm,-}+O_{H^{m-2}(\re^n)}(\jap{t}^{1-2\c})\,\, \text{for}\,\, \pm t>>1.
\end{align}
Moreover, for  $u\in \mathcal{X}^m$, $u\in \mathcal{X}_{\overline{F}}^m$ is equivalent to $g_{+,+}=g_{-,-}=0$.
\end{prop}

\begin{rem}
More generally, for  $u\in \mathcal{X}^m$, we have the following relations:
\begin{align*}
u\in \mathcal{X}_{F}^m\Leftrightarrow g_{+,-}=g_{-,+}=0,\,\, u\in \mathcal{X}_{out}^m\Leftrightarrow g_{-,+}=g_{-,-}=0,\,\, u\in \mathcal{X}_{in}^m\Leftrightarrow g_{+,+}=g_{+,-}=0.
\end{align*}
For the definition of $\mathcal{X}_{*}^m$ with $*\in \{F, out, in\}$, see \cite{GW1}.
\end{rem}

\begin{proof}
Set $u_*(t)=\begin{pmatrix}u(t)\\i^{-1}n\cdot \nabla u(t) \end{pmatrix}$.
By the Duhamel formula, we have
\begin{align*}
u_*(t)=\mathcal{U}(t,0)u_*(0)+i\mathcal{U}(t,0)\int_0^t\mathcal{U}(0,s)f(s)ds
\end{align*}
for large $|t|$ with some $f\in (\mathcal{Y}^m)^2$. Set
\begin{align*}
g_{\pm,*}:=\lim_{t\to \pm\infty}\mathcal{U}_{free}(0,t)u_*(t)= W_{\pm}^{-1}(u_*(0)+i\int_0^{\pm\infty}\mathcal{U}(0,s)f(s)ds)\in \mathcal{E}^m.
\end{align*}
Since $P$ is a second order perturbation of $P_0$, we have
\begin{align*}
\|u_*(t)-\mathcal{U}_{free}(t,0)g_{\pm,*}\|_{\mathcal{E}^{m-2}}=&\|\mathcal{U}_{free}(0,t)u_*(t)-g_{\pm,*}\|_{\mathcal{E}^{m-2}}\\
=&\|\int_{\pm \infty}^t\frac{d}{ds}(\mathcal{U}_{free}(0,s)u_*(s))ds\|_{\mathcal{E}^{m-2}}\\
=&O(\jap{t}^{1-2\c}),
\end{align*}
where we use $(\ref{asyp1})$ and $(\ref{asyp2})$.
The second assertion is proved by a simple calculation.
\end{proof}

\subsection{Radiation condition}\label{subsecrad}

We recall $\tilde{\x}=\frac{1}{2}\pa_{\x}p_0(\x)=(-\x_1,\x_2,..,\x_{n+1})=(-\t,\y)$ for $\x=(\t,\y)$. We define
\begin{align*}
\cos (x,x')=\frac{x\cdot x'}{|x||x'|}\quad \text{for}\quad x,x'\in \re^{n+1}.
\end{align*}
We also introduce the incoming/outgoing region for $P$ by
\begin{align}
&\O^0_{\e,R,in/out}:= \{(x,\x)\in \re^{2n+2}\mid |x|\geq R, \,\, \x\neq 0 \,\,  \pm\cos (x,\tilde{\x})<-1+\e\},\label{outindef}\\
&\O^0_{\e,R,m_0,in/out}:=\O_{\e,R,in/out}\cap \{ |p_0(\x)+m_0^2|<\e|\x|^2\}\nonumber.
\end{align}
We note $(x,\x)\in \O^0_{\e,R,m_0,in/out}\Rightarrow |\x|\geq r$ with some $r>0$.
In the next lemma, we shall show that the symbols of $D_t\mp \sqrt{-\Delta+m_0^2}$ are elliptic on the region $\O^0_{\e,R,in}\cap \{\mp t\geq 1\}$ respectively.

\begin{lem}\label{incomelliptic}
We use the notation $x=(t,y)\in \re\times \re^n$ and $\x=(\t,\y)\in \re\times \re^n$.
If $\e>0$ is small enough, then for $R\geq 1$ and $(x,\x)\in\O^0_{\e,R,m_0,in}\cap \{\mp t\geq 0\}$, we have
\begin{align}\label{incomttau}
\pm\t\leq -c_1|\x|,\quad  \pm t<-c_2|x| 
\end{align}
with constants $c_1,c_2>0$. In particular, we have
\begin{align}\label{inelliptices}
|\t\mp \sqrt{|\y|^2+m_0^2}|\geq c_3(1+|\x|)
\end{align}
for $(x,\x)\in\O^0_{\e,R,m_0,in}\cap \{\mp t\geq 0\}$ with a constant $c_3>0$.
\end{lem}

\begin{proof}
By a direct calculation, we have
\begin{align*} 
\frac{(1-\e)|\y|^2+m_0^2}{1+\e}\leq |\t|^2\leq \frac{(1+\e)|\y|^2+m_0^2}{1-\e},  \,\,  |\y|\leq \sqrt{\frac{1+\e}{2}}|\x|,\quad \text{for}\,\, (x,\x)\in\O^0_{\e,R,in/out}.
\end{align*}
Combining the above estimates with $\cos (x,\tilde{\x})<(-1+\e)$ and $|y|\leq |x|$, we obtain
\begin{align*}
-\t t<(-1+\e+\sqrt{\frac{1+\e}{2}})|x||\x|,\quad (\frac{1}{2}-\e)|\x|\leq |\t|\leq |\x|.
\end{align*}
Hence, if $\e>0$ is small enough, we have $-\t t<-\frac{1}{2}|x||\x|$ and $|\x|\leq 4|\t|\leq 4|\x|$. These inequality immediately imply $(\ref{incomttau})$. The inequality $(\ref{inelliptices})$ directly follows from $(\ref{incomttau})$.
\end{proof}

The next proposition says that the leading terms of $(\ref{waveexp})$ are microlocally negligible in the incoming region.

\begin{prop}\label{incomneg}
Let $g\in H^m(\re^{n}_y)$ and set $u_{\pm}(x)=u_{\pm}(t,y)=e^{\pm it\sqrt{-\Delta_y+m_0^2}}g(y)\in \mathcal{S}'(\re^{n+1})$. Let $\e>0$ be as in Lemma \ref{incomelliptic}. Then for $R\geq 1$ and for any $a_{\pm}\in S^{0,0}$ satisfying
\begin{align*}
\supp a_{\pm}\subset \O^0_{\e,R,m_0,in}\cap \{\mp t\geq 0\},
\end{align*}
we have $\Op(a_{\pm})u_{\pm}\in \mathcal{S}(\re^{n+1})$.
\end{prop}

\begin{proof}
We recall the notation $\x=(\t,\y)$ with $\t\in \re$ and $\y\in \re^n$.
Set $b_{\pm}(\x)=(\t\mp \sqrt{|\y|^2+m_0^2})(|\x|^2+1)^{-\frac{1}{2}}\in S^{0,0}$, which is elliptic on $\supp a_{\pm}$ by virtue of Lemma \ref{incomelliptic}. Then the standard elliptic parametrix and Borel's theorem show the existence of $c_{\pm}\in S^{0,0}$ and $r_{\pm}\in S^{-\infty,-\infty}$ such that
\begin{align*}
a_{\pm}=c_{\pm}\# b_{\pm}+r_{\pm}.
\end{align*}
On the other hand, a direct calculation gives $(D_t\mp \sqrt{-\Delta_y+m_0^2})u_{\pm}=0$. Hence, we have $\Op(b_{\pm})u_{\pm}=0$. Thus, we conclude $\Op(a_{\pm})u_{\pm}=\Op(r_{\pm})u\in \mathcal{S}(\re^{n+1})$.
\end{proof}

Set $\b(x,\x)=\cos (x,\pa_{\x}p(x,\x))$. By the long range condition, it turns out that $\b(x,\x)$ and $\cos(x,\tilde{\x})$ are sufficiently close for $|x|>>1$. Thus all the results hold if $\b$ is replaced by $\cos(x,\tilde{\x})$.
By Theorem \ref{Feyeq} $(i)$, Corollary \ref{microsmooth} and Lemma \ref{KGelliptic}, we have the following micorlocal radiation condition.

\begin{prop}[Radiation condition]\label{radcondprop}
Let $u\in \mathcal{S}'(\re^{n+1})$ and $f\in \mathcal{S}(\re^{n+1})$. Then there exists $R\geq 1$ and $\e>0$ such that the following holds: The identity $u=(P+m_0^2-i0)^{-1}f$ holds if and only if 
\begin{align*}
(P+m_0^2)u=f\quad \text{and}\quad \Op(a)u\in L^{2,-\frac{1}{2}+\d}(\re^{n+1})
\end{align*}
for any $a\in S^{0,0}$ supported on $\O^0_{\e,R,m_0,in}$ and for some $\d>0$. 
\end{prop}

Now we prove Theorem \ref{Feyeq} $(ii)$.

\begin{proof}[Proof of Theorem \ref{Feyeq} $(ii)$] Let $f\in \mathcal{S}(\re^{n+1})$ and set $u=(P+m_0^2)_{\overline{F}}^{-1}f\in \cap_{m\in \re}\mathcal{X}_{\overline{F}}^m$. By virtue of Proposition \ref{radcondprop}, it suffices to prove that $\Op(a)u\in L^{2,-\frac{1}{2}+\e}(\re^{n+1})
$ for any $a\in S^{0,0}$ supported on $\O^0_{\e,R,m_0,in}$ and for some $\e>0$. Using the estimate $|t|>c_2|x|$ for $(x,\x)\in \O^0_{\e,R,m_0,in/out}$ which follows from $(\ref{incomttau})$, we write
\begin{align*}
a=a_{+}+a_-\quad \supp a_{\mp}\subset \O^0_{\e,R,m_0,in}\cap \{\pm t\geq c_2|x|\}.
\end{align*}
By Proposition \ref{waveexprop} with $m=2$, we have
\begin{align*}
u(t)=e^{it\sqrt{-\Delta+m_0^2}}g_{\pm,+}+e^{-it\sqrt{-\Delta+m_0^2}}g_{\pm,-}+r_{\pm}(t)\quad \text{for}\quad \pm t>>1,
\end{align*}
where $g_{\pm}\in H^2(\re^n)$ with $g_{+,+}=g_{-,-}=0$ and $r_{\pm}= O_{L^2(\re^n)}(\jap{t}^{1-2\c})$. By the estimate $|t|>c_2|x|$ for $(x,\x)\in \O_{\e,R,m_0,in/out}$, we only need to prove 
\begin{align*}
\Op(a_{\mp})(e^{\mp it\sqrt{-\Delta+m_0^2}}g_{\pm,\mp}+r_{\pm})\in L^{2,-\frac{1}{2}+\d}(\re^{n+1})
\end{align*}
for some $\d>0$. Taking $\d$ small enough, we have $\Op(a_{\mp})r_{\pm}\in L^{2,-\frac{1}{2}+\d}(\re^{n+1})$, where we use $|t|>c_2|x|$ on $\supp a$ again. Moreover, Proposition \ref{incomneg} implies 
\begin{align*}
\Op(a_{\mp})(e^{\mp it\sqrt{-\Delta+m_0^2}}g_{\pm,\mp})\in \mathcal{S}(\re^{n+1})\subset L^{2,-\frac{1}{2}+\d}(\re^{n+1}).
\end{align*}
This completes the proof.

\end{proof}


\section{Mapping properties}

Main ingredients of the proof for Theorem \ref{Acompro} are the radial estimates proved in the Appendix and a less obvious subelliptic estimate.

\subsection{Subellipticity of the spectral parameter}
In this subsection, we establish a non-usual subelliptic estimates for $P-i$.
The proof of Lemma \ref{strange} below looks like the standard positive commutator argument at a first glance, however, it is very different from the usual positive commutator argument. In fact, in the proof below, we do not use the dynamical property for $H_p$ (the positivity of a symbol along the flow). We only need to use the symbol calculus and the symbol class $p\in S^{2,0}$. In this section, we assume that  $P$ is a self-adjoint operator on $L^2(\re^{n+1})$ of the form $(\ref{Pstd})$.

\begin{lem}\label{strange}
Let $k,l\in \re$, $z\in \mathbb{C}\setminus \re$. If $u\in H^{k+1/2, l-1/2}$ satisfies $(P-z)u\in H^{k,l}$, then we have $u\in H^{k,l}$.
\end{lem}

\begin{rem}
In this lemma, the assumption $z\notin \re$ is necessary. Moreover, this lemma holds for general pseudodifferential operator of the class $\Op S^{2,0}$.
\end{rem}

\begin{proof}
Take $\chi\in C_c^{\infty}(\re;\re)$ such that $\chi(t)=1$ on $|t|\leq 1$ and set
\begin{align*}
a_R(x,\x)=\chi(\frac{|(x,\x)|}{R}),\,\, \L_R=\jap{x}^l\jap{D}^k\Op(a_R)\in \Op S^{-\infty,-\infty}
\end{align*}
for $R\geq 1$.
We note $\L_R$ is uniformly bounded in $\Op S^{k,l}$. Since $\L_R\in \Op S^{-\infty,-\infty}$, we have
\begin{align*}
((P-z)u,\L_R^2u)_{L^2}-(\L_R^2u,(P-z)u)_{L^2}=2i\Im z\|\L_Ru\|_{L^2}^2+(u,[P,\L_R^2]u)_{L^2}
\end{align*}
for $u\in \mathcal{S}'(\re^{n+1})$. Now we let $u\in H^{k+1/2, l-1/2}$ satisfying $(P-z)u\in H^{k,l}$. Then we have
\begin{align*}
|(u,[P,\L_R^2]u)_{L^2}|\leq C\|u\|_{H^{k+1/2,l-1/2}}^2
\end{align*}
with a constant $C>0$ independent of $R\geq 1$. Moreover, by the Cauchy-Schwartz inequality, we have
\begin{align*}
|((P-z)u,\L_R^2u)_{L^2}-(\L_R^2u,(P-z)u)_{L^2}|\leq \frac{C}{|\Im z|}\|(P-z)u\|_{H^{k,l}}^2+|\Im z|\|\L_Ru\|_{L^2}^2
\end{align*}
with a constant $C>0$ independent of $R\geq 1$. Thus we have
\begin{align*}
|\Im z|\|\L_Ru\|_{L^2}^2\leq \frac{C}{|\Im z|}\|(P-z)u\|_{H^{k,l}}^2+C\|u\|_{H^{k+1/2,l-1/2}}^2.
\end{align*}
Using $\Im z\neq 0$ and a limiting procedure, we obtain $u\in H^{k,l}$.

\end{proof}

\subsection{Proof of Theorem \ref{Acompro}}

Now we prove Theorem \ref{Acompro}.
Suppose $z\in\mathbb{C}$. We only consider the case of $\Im z> 0$.
By duality, we only need to prove $P-z:\mathcal{S}(\re^{n+1})\to \mathcal{S}(\re^{n+1})$ is a homeomorphism.
Since the operator $P-z$ maps between $\mathcal{S}(\re^{n+1})$ continuously and is injective in $\mathcal{S}(\re^{n+1})$ (this follows from its essential self-adjointness on $\mathcal{S}(\re^{n+1})$), it suffices to prove that $P-z$ is surjective in $\mathcal{S}(\re^{n+1})$. In fact, if $P-z$ is bijective and continuous on $\mathcal{S}(\re^{n+1})$, then its inverse is also continuous by the open mapping theorem.

Let $f\in \mathcal{S}(\re^{n+1})$. Setting $u=(P-z)^{-1}f\in L^2(\re^{n+1})$, we have $(P-z)u=f$. Then Corollary \ref{microsmooth} implies $u\in \cap_{k\in \re,\d>0}H^{k,-\frac{1}{2}-\d}$. By Lemma \ref{strange}, we obtain $u\in \mathcal{S}(\re^{n+1})$. This means that $P-z$ is surjective in $\mathcal{S}(\re^{n+1})$.

\appendix

\section{Propagation estimates}

In this appendix, we give a proof for some propagation estimates, that is, the propagation of singularities and the radial estimates although these proofs may be well-known for specialists of the geometric scatering theory (see \cite[Appendix E.4]{DZ}). The radial estimates can be regarded as a microlocal alternative of the Mourre theory. However, its proof does not require technical assumptions such as the self-adjointness of the operator and the radial estimates often give an additional information such as the regularity of the function. We also mention that the abstract limiting absorption principle can be proved by a similar method in \cite{G}.

First, we shall explain key dynamical properties of $H_p$ for the propagation estimates.
We set
\begin{align*}
\b_0(x,\x)=\cos(x,\tilde{\x})=\frac{x\cdot \pa_{\x}p_0(\x)}{|x||\pa_{\x}p_0(\x)|},\quad \b(x,\x)=\frac{x\cdot \pa_{\x}p(x,\x)}{|x||\pa_{\x}p(x,\x)|}.
\end{align*}
A key property of the observable $\b_0$ is the following: Setting $\b_{\pm}=1\pm \b_0$, we have
\begin{align*}
&|\pa_{\x}p_0|^{-1}|x|H_{p_0}\b_{\pm}=\pm (1-\b_0^2)=\pm \b_{\mp}\b_{\pm}=\pm 2 \b_{\pm} \,\, \text{on}\,\,  L_{\mp}:=\{\b_{\pm}=0\},\\ 
&|\pa_{\x}p_0|^{-1}|x|H_{p_0}|x|^{-1}=-\b_0 |x|^{-1}=\pm |x|^{-1}\quad \text{on}\quad  L_{\mp}.
\end{align*}
Then it turns out that the sets $L_{\mp}=\{\b_{\pm}=0\}\cap\{|x|=\infty, \x\neq 0\}$ are attracting/repelling sets along the rescaled Hamiltonian $|x|H_{p_0}$. 
The sets $L_{\mp}\cap \{|x|=\infty,\x\neq 0\}$ are called the radial source/sink in \cite[DEFINITION E.50]{DZ} (see also \cite[Definition 2.3]{DW}), which are also called the incoming/outgoing regions in the scattering theory.

Similar to Section \ref{LAPsec}, in this appendix, let $P$ be a self-adjoint operator on the standard $L^2$-space $L^2(\re^{n+1})$ satisfying 
\begin{align*}
P=\Op(p)+\Op S^{1,-\m}.
\end{align*}
Since the multiplication operator $|g|^{\frac{1}{4}}$ preserves all spaces $H^{k,l}(\re^{n+1})$ with $k,l\in \re$, all the results hold for $P$ which is defined in $(\ref{Pdef})$.

\subsection{Preliminary}

Now we fix some notation.
For symbols $a,b$, we denote $a\Subset b$ if we have
\begin{align*}
\inf_{(x,\x)\in\supp a}|b(x,\x)|>0,
\end{align*}
and we denote $\Op(a)=:A\Subset B:=\Op(b)$ if $a\Subset b$. 

\begin{defn}
Let $k,l\in \re$.

\noindent$(i)$ We call $a\in S^{k,l}$ (or its quantization $\Op(a)$) is elliptic in a subset $\O\subset \re^{2n+2}$ if there exists $r\in S^{-\infty,-\infty}$ such that
\begin{align*}
\inf_{(x,\x)\in\O}\jap{x}^{-l}\jap{\x}^{-k}|a(x,\x)+r(x,\x)|>0.
\end{align*}

\noindent$(ii)$ We call $a\in S^{k,l}$ (or its quantization $\Op(a)$) is microlocally negligible outside a subset $\O\subset \re^{2n+2}$ if there exists $r\in S^{-\infty,-\infty}$ such that
\begin{align*}
\supp (a+r)\subset \O.
\end{align*}

\noindent$(iii)$  We call $a\in S^{k,l}$ (or its quantization $\Op(a)$) is microlocally contained in a pair $(\O_1,\O_2)$ if $a$ is elliptic in $\O_1$ and is microlocally negligible outside $\O_2$.

\noindent$(iv)$ Let $X\subset \mathcal{S}'(\re^{n+1})$ be a function space. For $u\in\mathcal{S}'(\re^{n+1})$, we say that $u\in X$ microlocally in a subset $\O\subset \re^{2n+2}$ if there exists $A=\Op(a)\in \Op S^{0,0}$ which is elliptic on $\O$ such that $Au\in X$.

\end{defn}

We set
\begin{align*}
\O_{\e,r,R,in/out}:=& \{(x,\x)\in \re^{2n+2}\mid |x|\geq R,\,\, |\x|\geq r, \,\,  \pm\b(x,\x)\leq -1+\e\},\\
\O_{r,R}(s):=& \{(x,\x)\in \re^{2n+2}\mid |x|\geq R,\,\, |\x|\geq r, \,\,  \b(x,\x)=s\},\\\O_{r,R}:=& \{(x,\x)\in \re^{2n+2}\mid |x|\geq R,\,\, |\x|\geq r\},\quad \O_{s_1,s_2,r,R,mid}=\cup_{s\in[s_1,s_2]}\O_{r,R}(s).
\end{align*}

\subsection*{Estimates for weight functions}

For $k,l\in \re$, $N,\k>0$ and $0\leq\d\leq 1$, we set
\begin{align*}
\l=\l_{k,l,\k,N,\d}=\jap{\x}^{k-\frac{1}{2}}\jap{\d\x}^{-|k|-N-1}\jap{x}^{l+\frac{1}{2}}\jap{\d x}^{-\k}.
\end{align*}

\begin{lem}\label{weighes}
\noindent $(i)$ For $r>0$, there exists $C_1>0$ independent of $0\leq\d\leq 1$ such that
\begin{align*}
H_p\l^2\leq C_1\frac{\jap{\x}}{\jap{x}}\l^2\quad \text{for}\quad x\in \re^n\quad \text{and}\quad |\x|\geq r.
\end{align*}

\noindent $(ii)$ Suppose $l>-\frac{1}{2}$ and $2l+1-2\k>0$. For $r>0$, $\e\in (0,1)$ and $R\geq 1$ large enough, there exists $C_2>0$ independent of $0\leq\d\leq 1$ such that
\begin{align*}
H_p\l^2\leq -C_2\frac{\jap{\x}}{\jap{x}}\l^2\quad \text{for}\quad (x,\x)\in \O_{\e,r,R,in}.
\end{align*}

\noindent $(iii)$ Suppose $l<-\frac{1}{2}$. For $r>0$, $\e\in (0,1)$ and $R\geq 1$ large enough, there exists $C_3>0$ independent of $0\leq\d\leq 1$ such that
\begin{align*}
H_p\l^2\leq -C_3\frac{\jap{\x}}{\jap{x}}\l^2\quad \text{for}\quad (x,\x)\in \O_{\e,r,R,out}.
\end{align*}

\end{lem}

\begin{proof}
First, we note $\pa_{x}p\in S^{2,-1-\m}$ and $\pa_{\x}p\in S^{1,0}$.
$(i)$ follows from a simple calculation. A simple calculation gives $H_p|\x|^2=O(|x|^{-1-\m}|\x|^3)$ and
\begin{align*}
&H_p\jap{x}^{2l+1}=(2l+1)|x||\pa_{\x}p(x,\x)|\b(x,\x)\jap{x}^{2l-1},\\
&H_p\jap{\d x}^{-2\k}=-2\k\d^2|x||\pa_{\x}p(x,\x)|\b(x,\x)\jap{\d x}^{-2\k-2}.
\end{align*}
Moreover, 
\begin{align*}
(\frac{2l+1}{\jap{x}^2}-\frac{\k\d^2}{\jap{\d x}^2})=&\frac{(2l+1)(1+\d^2|x|^2)-\k\d^2(1+|x|^2)}{\jap{x}^2\jap{\d x}^2}\\
&\begin{cases}
\geq c_1\jap{x}^{-2}\quad \text{if}\quad l>-\frac{1}{2},\,\, 2l+1-2\k>0\\
\leq -c_2\jap{x}^{-2}\quad \text{if}\quad l<-\frac{1}{2},
\end{cases}
\end{align*}
where $c_1,c_2>0$ are independent of $0<\d\leq 1$ and $x\in \re^n$.
Thus, for $R>0$ large enough, we have
\begin{align*}
H_p\l^2=&\jap{\x}^{2k-1}\jap{\d\x}^{-2|k|-2N-2}H_p( \jap{x}^{2l+1}\jap{\d x}^{-2\k})+O(\frac{|\x|}{\jap{x}^{1+\m}}\l^2)\\
=&(\frac{2l+1}{\jap{x}^2}-\frac{\k\d^2}{\jap{\d x}^2})|x||\pa_{\x}p(x,\x)|\b(x,\x)\l^2 +O(\frac{|\x|}{\jap{x}^{1+\m}}\l^2)\leq -C\frac{\jap{\x}}{\jap{x}}\l^2
\end{align*}
if $l>-\frac{1}{2}$, $2l+1-2\k>0$ and $(x,\x)\in \O_{\e,r,R,in}$ or if $l<-\frac{1}{2}$ and $(x,\x)\in \O_{\e,r,R,out}$. This proves $(ii)$ and $(iii)$.
\end{proof}

\subsection*{Estimates for cut-off functions}

First, we note that for $0<\e_0<1$ and $R\geq 1$ large enough, there exists $C_4>0$ such that 
\begin{align}\label{accposi}
H_p\b\geq C_4|\x|\jap{x}^{-1}\quad \text{for}\quad |x|\geq R,\,\, \x\neq 0,\,\, \b(x,\x)\in [-1+\e_0,1-\e_0]
\end{align}
since $H_p\b=|x|^{-1}|\pa_{\x}p(x,\x)|(1-\b^2)+O(\jap{x}^{-1-\m}|\x|)$.

\begin{lem}\label{cutoffes}

\noindent$(i)$ Let $-1<\b_1<\b_2<1$. For each $L>0$, $\e>0$ small enough, $R\geq 1$ large enough, there exist $a,b_1,b_2\in S^{0,0}$, $e\in S^{-\infty,-1/2}$ supported in $\{r/2\leq |\x|\leq \frac{5r}{2}\}$ satisfying the following properties: The symbol $a$ is microlocally contained in $(\O_{\b_1-\e,\b_2+\e,2r,2R,mid},\O_{\b_1-2\e,\b_2+2\e,r,R,mid})$, the symbols $b_1,b_2$ are microlocally negligible outside $\O_{\b_1-2\e,\b_1-\e,r,R,mid}$ and $\O_{r,R}$ respectively. Moreover, we have
\begin{align*}
H_pa^2\leq -L\frac{\jap{\x}}{\jap{x}}a^2+\frac{\jap{\x}}{\jap{x}}b_1^2+\frac{\jap{\x}}{\jap{x}}b_2^2 +\frac{\jap{\x}}{\jap{x}}e^2.
\end{align*}
In addition, if $\b_2<0$, we can take $b_2=0$.

\noindent$(ii)$ For $0<\e<\frac{1}{2}$, $r\geq 1$ and $R\geq 1$ large enough, there exist $a\in S^{0,0}$ which is microlocally contained in $(\O_{\e,2r,2R,in},\O_{2\e,r,R,in})$ and $e\in S^{-\infty,-1/2}$ supported in $\{r/2\leq |\x|\leq \frac{5r}{2}\}$ such that 
\begin{align*}
H_pa^2\leq \frac{\jap{\x}}{\jap{x}}e^2.
\end{align*}

\noindent$(iii)$ For $0<\e<\frac{1}{4}$, $r\geq 1$ and $R\geq 1$, there exist $a,b_1,b_2\in S^{0,0}$ and $e\in S^{-\infty,-1/2}$ supported in $\{r/2\leq |\x|\leq \frac{5r}{2}\}$ such that the following properties hold: The symbol $a$ is microlocally contained in $(\O_{\e,2r,2R,out},\O_{2\e,r,R,out})$, $b_1$ and $b_2$ are microlocally negligible outside  $\O_{1-2\e,1-\e,r,R,mid}$ and $\O_{r,R}$ respectively. Moreover, we have
\begin{align*}
H_pa^2\leq \frac{\jap{\x}}{\jap{x}}b_1^2+\frac{\jap{\x}}{\jap{x}}b_2^2 +\frac{\jap{\x}}{\jap{x}}e^2.
\end{align*}

\end{lem}

\begin{proof}
Take $\chi\in C^{\infty}(\re; [0,1])$ such that
\begin{align*}
\chi(t)=\begin{cases}
1\quad \text{for}\,\, t\leq 1,\\
0\quad \text{for}\,\, t\geq 2,
\end{cases}\quad \chi'(t)\leq 0.
\end{align*}
Moreover, we set $\bar{\chi}_R(x)=1-\chi(|x|/R)$ for $R>0$.

\noindent$(i)$  Take $\e>0$ small enough and $\rho_{mid}\in C^{\infty}(\re;[0,1])$ such that 
\begin{align*}
\rho_{mid}(t)=\begin{cases}
1 \quad \text{for}\quad \b_1-\e\leq t\leq \b_2+\e\\
0 \quad \text{for}\quad t\leq \b_1-2\e\,\,\text{or}\,\, \b_2+2\e\leq t
\end{cases}\,\, \rho_{mid}'(t)\leq 0\,\,\text{for}\,\, \b_2+\e\leq t\leq \b_2+2\e.
\end{align*}
Moreover, take $R\geq 1$ large enough such that $H_p\b\geq C_0|\x|\jap{x}^{-1}$ for $(x,\x)\in \supp \rho_{mid}(\b)$ with $|x|\geq R$ and take $M>0$ such that
\begin{align*}
H_pe^{-M\b}(x,\x)\leq -L\frac{\jap{\x}}{\jap{x}} e^{-M\b}\quad \text{for}\quad (x,\x)\in \O_{\b_1-2\e,\b_2+2\e,r,R,mid}.
\end{align*}
Now we set
\begin{align*}
a(x,\x)=e^{-M\b(x,\x)}\rho_{mid}(\b(x,\x))\overline{\chi}_R(x)\overline{\chi}_r(\x)\in S^{0,0}.
\end{align*}
Then we have $\supp a\subset \O_{\b_1-2\e,\b_2+2\e,r,R,mid}$ and
\begin{align*}
H_pa^2\leq -L\frac{\jap{\x}}{\jap{x}}a^2+e^{-2M\b}(x,\x)H_p(\rho_{mid}(\b)^2\overline{\chi}_R^2\overline{\chi}_r^2).
\end{align*}
Since $H_p\b(x,\x)\geq 0$ on $\O_{\b_1-2\e,\b_2+2\e,r,R,mid}$ and $\rho_{mid}'(t)\leq 0$ for $\b_2+\e\leq t\leq \b_2+2\e$, there exists $\rho_1\in C^{\infty}(\re;\re)$ supported in $[\b_1-3\e,\b_1-\e/2]$ such that
\begin{align*}
H_p\rho_{mid}(\b)^2\leq C\frac{\jap{\x}}{\jap{x}}\rho_1(\b)^2\quad \text{for}\quad (x,\x)\in\O_{r,R}.
\end{align*}
Now we construct $b_1=e^{-M\b}\rho_1(\b)\overline{\chi}_R\overline{\chi}_r$. The symbols $b_2$ and $e$ are similarly constructed, where we note that $b_2$ and $e$ come from the term $H_p(\overline{\chi}_R)$ and $H_p(\overline{\chi}_r)$ respectively. If $\b_2<0$ and $\e>0$ small enough, then $H_p(\overline{\chi}_R)\leq 0$ and hence we can take $b_2=0$.

\noindent$(ii)$ Take $\rho_{in}\in C^{\infty}(\re;[0,1])$ such that $\rho_{in}(t)=1$ on $t\leq -1+\e$, $\rho_{in}(t)=0$ for $t\geq -1+2\e$ and $\rho_{in}'(t)\leq 0$. We set
\begin{align*}
a(x,\x)=\rho_{in}(\b(x,\x))\overline{\chi}_R(x)\overline{\chi}_r(\x)\in S^{0,0}.
\end{align*}
By virtue of $(\ref{accposi})$, we can take $R\geq 1$ large enough such that $H_p\b\geq 0$ if $-1+\e\leq \b(x,\x)\leq -1+2\e$ and $|x|\geq R$. Then we have
\begin{align*}
H_p(\rho_{in}(\b))=(H_p\b)\rho_{in}'(\b)\leq 0,\quad H_p\overline{\chi}_R=-R^{-1}\chi'(|x|)|\pa_{\x}p|\b\leq 0
\end{align*}
for $|x|\geq R$ and $(x,\x)\in \supp \rho_{in}(\b)$. Then the symbol $e$ can be constructed associated with the term $a\rho_{in}\overline{\chi}_RH_p\overline{\chi}_r$ as in $(i)$.

\noindent$(iii)$ Take $\rho_{out}\in C^{\infty}(\re;[0,1])$ such that $\rho_{out}(t)=1$ on $t\geq 1-\e$ and $\rho_{out}(t)=0$ for $t\leq 1-2\e$.  We set
\begin{align*}
a(x,\x)=\rho_{out}(\b(x,\x))\overline{\chi}_R(x)\overline{\chi}_r(\x)\in S^{0,0}.
\end{align*}
Then our claim follows as in $(i)$ and $(ii)$, where we note that the term $b_1$ comes from $H_p(\rho_{in}(\b))$ and the term $b_2$ comes from the term $H_p(\overline{\chi}_R)$ as in $(i)$. 

\end{proof}

\subsection{Radial estimates}

In this subsection, we do not assume the null non-trapping condition (Assumption \ref{nulltrapp}). In the other parts of this appendix, we shall impose Assumption \ref{nulltrapp}.

Let $a,b_1,b_2,e\in S^{0,0}$ be as in Lemma \ref{cutoffes} $(i)$, $(ii)$ and $(iii)$ respectively, where we take $b_1=b_2=0$ in the case $(ii)$. Set $A=\Op(a)$, $B_j=\Op(b_j)$ and $E=\Op(e)$ and take $A\in \Op S^{0,0}$ such that 
\begin{align*}
A, B_j, E\Subset A'.
\end{align*}
Correspondingly, we have the following theorem.

\begin{thm}\label{propathm} Let 
\begin{align}\label{radzrange}
z\in \{z\in \mathbb{C}\mid \Im z\geq  0\}:=\mathbb{C}_{+}.
\end{align}
We consider the estimate
\begin{align}
\|Au\|_{H^{k,l}}\leq& C\|A'(P-z)u\|_{H^{k-1, l+1}}+C\|B_1u\|_{H^{k,l}}\nonumber\\
&+C\|B_2u\|_{H^{k,l}}+C\|Eu\|_{H^{k,l}}+C\|u\|_{H^{-N,-N}}.\label{radineq}
\end{align}
We have the following statements.

\noindent$(i)$$($Propagation of singularity$)$ Let $k,l\in \re$, $N>0$ and $-1<\b_1<\b_2<1$. Suppose that $u\in H^{-N,-N}\cap H^k_{loc}$ satisfies $u\in H^{k,l}$ microlocally on $\{|\x|\leq 3r\}$ and
\begin{align*}
&u\in H^{k,l}, \quad \text{microlocally on}\quad \O_{r,R}(\b_2)\\
&(P-z)u\in H^{k+1,l-1}\quad \text{microlocally on}\quad \O_{\b_1-2\e,\b_2+2\e,r,R,mid}
\end{align*}
for $r>0$, $\e>0$ small enough and $R\geq 1$ large enough. Then we have $u\in H^{k,l}$ microlocally on $\O_{r',R'}(\b_1)$ for some $r',R'>0$.

 More precisely, the following statement holds:
Then, for $u\in\mathcal{S}'(\re^{n+1})$ such that the right hand side of $(\ref{radineq})$ is bounded, the estimate $(\ref{radineq})$ hold with a constant $C>0$ independent of $z\in \mathbb{C}_+$.

\noindent$(ii)$$($Radial source estimate$)$ Let $l>-1/2$, $N>0$ and $k\in \re$. Suppose that $u\in\mathcal{S}'(\re^{n+1})$ satisfies $u\in H^{k,l}$ microlocally on $\{|\x|\leq 3r\}$ and
\begin{align*}
u\in H^{-N,l_0}\quad \text{and}\quad (P-z)u\in H^{k+1,l-1}\quad \text{microlocally on}\quad \O_{r,R}(-1)
\end{align*}
for some $l_0>-1/2$, $r>0$ and $R\geq 1$ large enough. Then we have $u\in H^{k,l}$ microlocally on $\O_{r',R'}(-1)$ for some $r',R'>0$.

More precisely, the following statement holds: Set $B_1=B_2=0$.
Consider $u\in\mathcal{S}'(\re^{n+1})$ such that the right hand side of $(\ref{radineq})$ is bounded and $A'u\in H^{-N,l_0}$. Then the estimate $(\ref{radineq})$ holds for such $u$ with a constant $C>0$ independent of $z\in \mathbb{C}_+$.

\noindent$(iii)$$($Radial sink estimate$)$ Let $l<-\frac{1}{2}$, $N>0$ and $k\in \re$. Suppose that $u\in H^{-N,-N}\cap H^k_{loc}$ satisfies $u\in H^{k,l}$ microlocally on $\{|\x|\leq 3r\}$ and
\begin{align*}
&u\in H^{k,l}, \quad \text{microlocally on}\quad \O_{r,R}(1-\e)\\
&(P-z)u\in H^{k+1,l-1}\quad \text{microlocally on}\quad \O_{1-2\e,1,r,R,mid}
\end{align*}
for $r>0$, $0<\e<1/2$ and $R\geq 1$ large enough. Then we have $u\in H^{k,l}$ microlocally on $\O_{r',R'}(1)$ for some $r',R'>0$.

 More precisely, the following statement holds:
Then, for $u\in\mathcal{S}'(\re^{n+1})$ such that the right hand side of $(\ref{radineq})$ is bounded, the estimate $(\ref{radineq})$ holds for such $u$ with a constant $C>0$ independent of $z\in \mathbb{C}_+$.

\end{thm}

\begin{rem}
For the radial source estimate $(ii)$, we do not need the regularity assumption $u\in H^k_{loc}$, which is different from the propagation of singularity $(i)$ and the radial sink estimate $(iii)$.
\end{rem}

\begin{rem}
If $z\neq 0$ and if we take $r>0$ sufficiently small, then we have $|p(x,\x)-z|\geq C>0$ on $\{|\x|\leq 3r\}$.
This implies that under $z\neq 0$, the assumption $u\in H^{k,l}$ microlocally on $\{|\x|\leq 3r\}$ can be removed if we also assume $A'$ is elliptic on $\{|\x|\leq 3r\}$. 
\end{rem}

\begin{rem}
For $z\in\mathbb{C}_-=\{w\in\mathbb{C}\mid \Im w\leq 0\}$, similar results hold although the direction of the propagation should be reversed.
\end{rem}

\subsection*{Commutator estimates}

For the proof of Theorem \ref{propathm}, the following commutator calculus has an important role: For pseudodifferential operators $A, \L$, where $A$ is formally self-adjoint and $\Im z\geq 0$, we have
\begin{align}\label{KDcomeq} 
\Im((P-z)u, A\L^*\L Au)_{L^2}=&-(u,[P,iA\L^* \L A]u)_{L^2}+\Im z\|\L Au\|_{L^2}^2
\end{align}
for $u\in \mathcal{S}(\re^{n+1})$. Moreover, the equation $(\ref{KDcomeq})$ with the Cauchy-Schwartz inequality implies that for any small $\e_1>0$, there exists $C>0$ such that
\begin{align}\label{KDcomeq2}
-(u,[P,iA\L^*\L A]u)_{L^2}\leq C\|\L A(P-z)u\|_{H^{-\frac{1}{2},\frac{1}{2}}}^2+\e_1\|\L Au\|_{H^{\frac{1}{2},-\frac{1}{2}}}^2.
\end{align}

We set
\begin{align*}
\L=\L_{k,l,\k,N,\d}=\Op(\l_{k,l,\k,N,\d}),\quad\,\, \Theta=\jap{x}^{-\frac{1}{2}}\jap{D}^{\frac{1}{2}},\quad \L_{\d}=\jap{\d x}^{-\k}\jap{\d D}^{-|k|-N-1}.
\end{align*}

\begin{proof}[Proof of Theorem \ref{propathm}]
We may assume $N>0$ is sufficiently large.
If we take $R\geq 1$ large enough and $\e>0$ small enough, Lemma \ref{weighes} and Lemma \ref{cutoffes} yield
\begin{align*}
H_p(a^2\l^2)\leq (-C_4\frac{\jap{\x}}{\jap{x}}a^2+\frac{\jap{\x}}{\jap{x}}b_1^2+\frac{\jap{\x}}{\jap{x}}b_2^2 +\frac{\jap{\x}}{\jap{x}}e^2)\l^2
\end{align*}
with a constant $C_4>0$, where we take $L=-2C_0$ in the case $(i)$. The sharp G\aa rding inequality (Lemma \ref{prelimlem} $(iii)$) gives
\begin{align*}
C_4A\L_{\d}\L\Theta^2\L\L_{\d} A\leq&-[P,iA\L_{\d}\L^2\L_{\d}A] +\sum_{j=1}^2B_j\L_{\d}\L\Theta^2\L \L_{\d} B_j+E\L_{\d}\L\Theta^2\L \L_{\d}E\\
&+A_1\L_{\d}\L\Theta\jap{x}^{-\frac{1}{2}}\jap{D}^{-1}\jap{x}^{-\frac{1}{2}} \Theta\L \L_{\d}A_1
\end{align*}
up to the term $\Op S^{-\infty,-\infty}$ uniformly bounded in $0\leq\d\leq 1$.  Here, the last term comes from a remainder term of the sharp G\aa rding inequality and from $[P-\Op(p),iA\L_{\d}\L^2\L_{\d}A]$. Moreover, $A_1\in \Op S^{0,0}$ is elliptic in
\begin{align*}
\supp a\cup \supp b_1\cup \supp b_2\cup \supp e.
\end{align*}
Now $(\ref{KDcomeq})$ and $(\ref{KDcomeq2})$ imply that for $u\in \mathcal{S}(\re^{n+1})$,
\begin{align*}
\|\L_{\d} Au\|_{H^{k, l}}^2\leq C\|\L_{\d} A(P-z)u\|_{H^{k-1, l+1}}^2+C\|\L_{\d}B_1u\|_{H^{k,l}}^2+C\|\L_{\d}B_2u\|_{H^{k,l}}^2\nonumber\\
+C\|\L_{\d}Eu\|_{H^{k,l}}^2+C\|\L_{\d}A_1u\|_{H^{k-\frac{1}{2},l-\frac{1}{2}}}^2+C\|u\|_{H^{-N,-N}}^2
\end{align*}
with a constant $C>0$ independent of $0\leq\d\leq 1$. Now it suffices to prove the lower order term $\|\L_{\d}A_1u\|_{H^{k-\frac{1}{2},l-\frac{1}{2}}}$ and relax the a priori regularity assumption of $u$.

By a standard bootstrap argument, the term $\|\L_{\d}A_1u\|_{H^{k-\frac{1}{2},l-\frac{1}{2}}}^2$ can be absorbed into the term $\|u\|_{H^{-N,-N}}^2$ for the cases $(i)$ and $(iii)$. For the case $(ii)$, using the bootstrap argument, an interpolation argument, and the part $(i)$, the term $\|\L_{\d}A_1u\|_{H^{k-\frac{1}{2},l-\frac{1}{2}}}^2$ can be absorbed into the left hand side and the term $\|u\|_{H^{-N,-N}}^2$. For a detail, see \cite[Part 2 of proof for Theorem E.52]{DZ}. Thus, for $u\in\mathcal{S}(\re^{n+1})$, we have
\begin{align}
\|\L_{\d} Au\|_{H^{k, l}}^2\leq C\|\L_{\d} A(P-z)u\|_{H^{k-1, l+1}}^2+C\|\L_{\d}B_1u\|_{H^{k,l}}^2+C\|\L_{\d}B_2u\|_{H^{k,l}}^2\nonumber\\
+C\|\L_{\d}Eu\|_{H^{k,l}}^2+C\|u\|_{H^{-N,-N}}^2.\label{sourceapproxi2}
\end{align}

Taking $\d=0$ and using a standard approximation argument with Lemma \ref{convergence}, we have $(\ref{sourceapproxi2})$ for $u\in \mathcal{S}'(\re^{n+1})$ satisfying $A'u\in H^{k,l}$ and $A'(P-z)u\in H^{k-1,l+1}$. In fact, take $a\in C_c^{\infty}(\re^{2n+2})$ with $a(x,\x)=1$ for $|(x,\x)|\leq 1$ and set $a_R(x,\x)=a(x/R,\x/R)$. Substituting $\Op(a_R)u$ into $(\ref{sourceapproxi2})$ with $\d=0$ and taking $R\to \infty$, we obtain $(\ref{sourceapproxi2})$ for such $u$ by Lemma \ref{convergence}.
 
Finally, we consider $u\in\mathcal{S}'(\re^{n+1})$ such that the right hand side of $(\ref{radineq})$ is bounded. For the case $(ii)$, we also assume $A'u\in H^{-N,l_0}$. We recall the notation 
\begin{align*}
\L_{\d}=\jap{\d x}^{-\k}\jap{\d D}^{-|k|-N-1}.
\end{align*}
First, we consider the cases $(i)$ and $(iii)$. Take $\k$ and $N$ large enough. Substituting $\Op(a_R)u$ into $(\ref{sourceapproxi2})$ with $0<\d\leq 1$, where $a_R$ is as above, and taking $R\to \infty$, we obtain $(\ref{sourceapproxi2})$, which implies $(\ref{radineq})$ for such $u$ by taking $\d\to 0$.
For the case $(ii)$, we can use this procedure for arbitrary $N>0$ and for $\k>-l-\frac{1}{2}$ due to Lemma \ref{weighes}. Thus we need the additional assumption  $A'u\in H^{-N,l_0}$. We omit the detail.
For a similar argument, see \cite[Exercises E.31, E.35, E.36]{DZ}.

\end{proof}

\subsection{Propagation to the radial source in the past infinity}

In order to control the regularity for a bounded region of the $x$-space, we use the standard propagation of singularity theorem and the null non-trapping condition. To apply it, we need the following dynamical result.

\begin{lem}
Let $(x_0,\x_0)\in T^*\re^{n+1}$ with $\x_0\neq 0$ and $p(x_0,\x_0)=0$. We denote $z(t)=z(t,x_0,\x_0)$, $\z(t)=\z(t,x_0,\x_0)$ and $\b(t)=\b_0(z(t),\z(t))$. Then for any $0<\e<1$ and $R\geq 1$, there exists $T>0$ such that $|z(-T)|>R$ and
\begin{align}\label{dynamical}
\b(-T)<(-1+\e).
\end{align}

\end{lem}

\begin{rem}
The argument below is standard. However, for its justification, we need some estimates for the classical trajectories which are proved in \cite[Appendix A]{NT}.
\end{rem}

\begin{proof}
Let $0<\e<1$ and $R\geq 1$. Take $R_0\geq R$ such that
\begin{align}\label{cltjpr}
\begin{cases}
H_p^2|x|^2\geq C|\x|^2\quad \text{if}\quad |x|\geq R_0\\
H_p\b(x,\x)\geq C\jap{x}^{-1}|\x|\quad \text{if}\quad |x|\geq R_0\,\, \text{and}\,\, (-1+\e)\leq \b(x,\x)\leq 0. 
\end{cases}
\end{align}
By Assumption \ref{nulltrapp}, we can choose $T_0>0$ such that
\begin{align*}
|z(-T_0)|\geq 2R_0,\,\, \frac{d}{dt}|z(t)|^2|_{t=-T_0}\leq 0.
\end{align*}
This inequality and $(\ref{cltjpr})$ imply that $|z(-t)|\geq R_0$ for $t\geq T_0$. 
By the proof of \cite[Lemma A.2]{NT} and \cite[Corollary A.4]{NT}, we have
\begin{align*}
C_1^{-1}\jap{t}\leq  \jap{z(t)}\leq C_1\jap{t}\quad C_2^{-1}\leq |\z(t)|\leq C_2
\end{align*}
for all $t\in \re$ with a constant $C_1,C_2>0$. Now suppose that $(\ref{dynamical})$ fails. 
By the inequality $(\ref{cltjpr})$, we obtain
\begin{align*}
\b(-T_0)=\b(-t)+\int_{-t}^{-T_0}\b'(s)ds\geq-1+ CC_1^{-1}C_2^{-1}\int_{-t}^{-T_0}\jap{s}^{-1}ds=\infty\quad \text{as}\quad t\to \infty
\end{align*}
which is a contradiction.  

\end{proof}

Combining the radial source estimate with the standard propagation of singularity, we have the following corollary which is a generalization of \cite[Proposition 3.2]{NT}.

\begin{cor}\label{KDlocsmooth}
Let $k\in \re$, $\d>0$ and $z\in \mathbb{C}_{+}\setminus \{0\}$. Suppose that 
\begin{align*}
u\in H^{k,-\frac{1}{2}+\d}(\re^{n+1}) \quad \text{microlocally on}\,\, \O_{r,R}(-1)   \quad \text{and}\quad (P-z)u\in \mathcal{S}(\re^{n+1})
\end{align*} 
with $R\geq 1$ large enough and $r>0$ small enough.
Then we have $u\in C^{\infty}(\re^{n+1})$.
\end{cor}

\begin{proof}
We shall prove $u\in C^{\infty}(\re^{n+1})$ microlocally near $(x_0,\x_0)\in \re^{n+1}\setminus \{\x=0\}$.
The last lemma implies that for $R\geq 1$ large enough, and $0<\e<1$, there exists $(x_1,\x_1)\in T^*\re^{n+1}\setminus 0$ such that $(x_1,\x_1)$ lies in the same integral curve of $H_p$ and 
\begin{align*}
|x_1|>R,\,\, \b(x_1,\x_1)<-1+\e.
\end{align*}
Then it suffices to prove $u\in C^{\infty}(\re^{n+1})$ microlocally near $(x_1,\x_1)\in \re^{n+1}\setminus \{\x=0\}$ by the standard propagation of singularities theorem. Moreover, since $\x_0\neq 0$ and since $p$ is homogeneous of degree $2$, we have $\x_1\neq 0$. Moreover, we may assume $|\x_1|$ is large enough since the wave front set is invariant under a scaling with respect to the $\x$-variable. To use Theorem \ref{propathm} $(ii)$, we shall check $u\in \mathcal{S}(\re^{n+1})$ microlocally in $\{|\x|\leq 3r\}$. We note that if $z\neq 0$, then $P-z$ is elliptic on $\{|\x|\leq 4r\}$ with some $r>0$. Then the standard elliptic parametrix construction yields $u\in \mathcal{S}(\re^{n+1})$ microlocally in $\{|\x|\leq 3r\}$. Now  Theorem \ref{propathm} $(ii)$ implies that for  $R>>1$ large enough and $0<\e<1$, we have $u\in C^{\infty}(\re^{n+1})$ microlocally on the incoming region $\O_{\e,r,R,in}$. This completes the proof since $(x_1,\x_1)\in \O_{\e,r,R,in}$.
\end{proof}

Summarizing the above results, we obtain the following corollary.

\begin{cor}\label{microsmooth}
Let $z\in \mathbb{C}_{+}\setminus \{0\}$, $k\in \re$ and $\d>0$. Suppose $u\in\mathcal{S}'(\re^{n+1})$ satisfies $(P-z)u\in \mathcal{S}(\re^{n+1})$ and $u\in H^{k,-\frac{1}{2}+\d}(\re^{n+1})$ microlocally in the incoming region $\O_{r,R}(-1)$ with $R\geq 1$ large enough and $r>0$ small enough. Then we have $u\in \cap_{k\in \re, \d>0}H^{k,-\frac{1}{2}-\d}$. Moreover, we have $u\in \mathcal{S}(\re^{n+1})$ microlocally away from the outgoing region $\O_{r,R}(1)$ for $R\geq 1$ large enough and $r>0$ small enough..
\end{cor}

\subsection{Elliptic estimate}

The following lemma is proved by a standard parametrix construction.

\begin{lem}\label{KGelliptic}
Let $m_0>0$. If $u\in \mathcal{S}'(\re^{n+1})$ satisfies $(P+m_0^2)u\in H^{k,l}$ with $k,l\in \re$, then we have $u\in H^{k+2,l}$ microlocally in $\{(x,\x)\in \re^{2n+2}\mid |x|\geq R,\,\, |p_0(\x)+m_0^2|\geq \e|\x|^2\}$ for $R\geq 1$ large enough and $\e>0$.
\end{lem}

\subsection{Absence of resonances}

By the proof of \cite[Proposition 7]{V}, we have the following proposition.

\begin{prop}\label{abres}
Let $z\in\mathbb{C}_{+}\setminus \{0\}$. Suppose that a distribution $u\in \mathcal{S}'(\re^{n+1})$ satisfies $(P-z)u=0$ and $u\in \mathcal{S}(\re^{n+1})$ microlocally away from the outgoing region $\O_{r,R}(1)$ for some $R\geq 1$ large enough and $r>0$ small enough. Then we have $u\in \mathcal{S}(\re^{n+1})$.
\end{prop}

\end{document}